\documentclass[journal]{IEEEtran}

\usepackage{amssymb}
\usepackage{amsmath}
\usepackage{amsfonts}
\usepackage{multirow}
\usepackage{graphicx}
\usepackage{textcomp}
\usepackage{amsthm}
\renewcommand{\vec}[1]{\boldsymbol{#1}}
\usepackage{setspace}
\usepackage{subfigure}
\usepackage{algorithm}
\usepackage{algorithmic}
\usepackage{color}

\ifCLASSINFOpdf
\else
\fi

\newtheorem{thm}{Theorem}
\newtheorem{lem}{Lemma}

\newtheorem{defn}{Definition}

\newtheorem{rem}{Remark}

\begin{document}

\title{Combinatorial Resources Auction in Decentralized Edge-Thing Systems Using Blockchain and Differential Privacy}

\author{Jianxiong Guo,
	Xingjian Ding,
	and Weijia Jia,~\IEEEmembership{Fellow,~IEEE}
	\thanks{J. Guo and W. Jia are with the BNU-UIC Institute of Artificial Intelligence and Future Networks, Beijing Normal University at Zhuhai, Zhuhai, Guangdong 519087, China, and also with the Guangdong Key Lab of AI and Multi-Modal Data Processing, BNU-HKBU United International College, Zhuhai, Guangdong 519087, China. (e-mail: jianxiongguo@bnu.edu.cn; jiawj@uic.edu.cn)
		
	X. Ding are with the School of Software Engineering, Beijing University of Technology, Beijing 100124, China. (e-mail: dxj@bjut.edu.cn)
	
	\textit{(Corresponding author: Jianxiong Guo.)}
	}% <-this 
	\thanks{Manuscript received April 19, 2005; revised August 26, 2015.}}

\markboth{Journal of \LaTeX\ Class Files,~Vol.~14, No.~8, August~2015}%
{Shell \MakeLowercase{\textit{et al.}}: Bare Demo of IEEEtran.cls for IEEE Journals}

\maketitle

\begin{abstract}
	With the continuous expansion of Internet of Things (IoT) devices, edge computing mode has emerged in recent years to overcome the shortcomings of traditional cloud computing mode, such as high delay, network congestion, and large resource consumption. Thus, edge-thing systems will replace the classic cloud-thing/cloud-edge-thing systems and become mainstream gradually, where IoT devices can offload their tasks to neighboring edge nodes. A common problem is how to utilize edge computing resources. For the sake of fairness, double auction can be used in the edge-thing system to achieve an effective resource allocation and pricing mechanism. Due to the lack of third-party management agencies and mutual distrust between nodes, in our edge-thing systems, we introduce blockchains to prevent malicious nodes from tampering with transaction records and smart contracts to act as an auctioneer to realize resources auction. Since the auction results stored in this blockchain-based system are transparent, they are threatened with inference attacks. Thus in this paper, we design a differentially private combinatorial double auction mechanism by exploring the exponential mechanism such that maximizing the revenue of edge computing platform, in which each IoT device requests a resource bundle and edge nodes compete with each other to provide resources. It can not only guarantee approximate truthfulness and high revenue, but also ensure privacy security. Through necessary theoretical analysis and numerical simulations, the effectiveness of our proposed mechanisms can be validated.
\end{abstract}


\IEEEpeerreviewmaketitle

\section{Introduction}
\IEEEPARstart{W}{ith} the rapid improvement of electronic equipment and communication infrastructure, Internet of Things (IoT) has become a hot research topic connecting the physical environment to the syberspace system. IoT devices are ubiquitous and play an important role in our lives, such as mobile phones, cemeras, automobiles, and traffic sensors. In recent years, the number of IoT devices has exploded. Based on a survey conducted by Cisco \cite{ni2017securing} \cite{alli2019secoff}, they predicted that more than 75 billion IoT devices would go into operation before 2025. These IoT devices produce a large amount of data. How to utilize these data to better serve the society has attracted more and more attention in academia and industry, and has driven a series of downstream industries such as smart home, smart supply chain, healthcare, and product traceability.

It is not easy to process these data to produce valuable information, which usually involves some artificial intelligence algorithms or data mining techniques. It requires IoT devices to have a certain amount of computing power and storage space. However, most IoT devices are lightweight, which can only temporarily store a small amount of data and perform simple operations. In the traditional cloud-thing system, IoT devices rely on the computing power, network bandwidth, and storage space of cloud centers to implement their own functions. Usually, cloud centers are far away from IoT devices, which leads to high energy consumption and network deley. In addition, it also faces the threat of the single point of failure \cite{aitzhan2016security} \cite{guo2020blockchain}, making this system more unreliable. As a result, the cloud-edge-thing system came into being. There are a lot of edge servers distributed in every corner of the space evenly. These edge nodes provide nearby IoT devices with the resources they need. Thus, IoT devices can offload their tasks to neighboring edge nodes instead of cloud centers. Although it overcomes some defects, especially long distance transmission, in cloud-thing systems, the cloud-edge-thing system does not get rid of the control of cloud centers completely.

Therefore, we focus on the edge-thing system in this paper, which is completely decentralized without the management of a third-party authority. But in the resources allocation between IoT devices and edge nodes, they do not trust each other due the conflict of interests, in which both entities want to maximize their own revenues. Moreover, the transaction records stored in edge nodes may be maliciously tampered with. With this in mind, blockchain \cite{nakamoto2008bitcoin} is an opportunity to provide a secure peer-to-peer (P2P) network. Blockchain is a distributed ledger for storing real-time data generated by all active participants in the system. It can not only achieve complete decentralization, but also has the characteristics of tampering-proof and transparency. The secure P2P network proved by blockchain can be used as a supplementary technique to design our edge-thing system.

In order to reflect the real market fluctuation and relationship between supply and demand, auction has been proven to be effective, so that IoT devices can get the resources they need at acceptable prices and edge nodes can benefit from providing resources. In this paper, we design a blockchain-based edge-thing system, where the resources allocation and pricing are realized by a combinatorial double auction mechanism. Here, IoT devices are buyers requesting resources and edge nodes are sellers providing resources. Because we have adopted a completely decentralized architecture, there is no suitable entity to act as an auctioneer responsible for executing the auction mechanism and deciding auction results. In our system, the auction mechanism is stored in the smart contract that is built in the blockchain, which can be run automatically when receiving all requests from IoT devices and edge nodes. Different IoT devices have different requirements for each resource type. For example, a device in smart home needs more computing power to implement intelligent algorithms, but a traffic monitor needs more storage space to store road condition data. Each IoT device ususally request a bundle of resources according to its task and gives a total bid, which is the reason for the formulation of a combinatorial double auction. The core of designing an auction mechanism is to easure the truthfulness, so as to encourage buyers/sellers to bid/ask their true valuations. 

Since allocation and pricing results are transparent in the blockchain, it exsits possible risk of exposing bids/asks of buyers/sellers. The bidding/asking information is their privacy, which may contain some commercial secrets. Adversaries could infer others' bids/asks through comparing the public auction results in multiple rounds by changing its bid/ask. This is known as ``inference attack'' \cite{zhu2014differentially} \cite{zhu2015differentially}. In order to prevent players from being trouble by inference attacks, differential privacy \cite{dwork2008differential} is a promising technology with strong thoeretical guaranatees that can be introduced in designing auction mechanisms. Even though several differential privacy-based auction mechanisms have been proposed in previous literature \cite{mcsherry2007mechanism} \cite{zhu2014differentially} \cite{jin2016enabling} \cite{lin2016bidguard} \cite{guo2021differential} \cite{ni2021differentially}, they are very different from the auction in our edge-thing systems. First, our auction is combinatorial because every buyer gives a total bid for a bundle of resources. Second, each edge nodes can only provide a limited amount of resource for each resource type. Third, the resource request of an IoT device can only be satisfied by one edge node, and the distance between them is constrained. Consider the real situation in edge-thing systems, we design a differentially private combinatorial double auction mechanism by exploring the exponential mechanism that selects the final pricing with a probability proportional to its corresponding revenue. On the premise of ensuring that the privacy is not exposed, it achieves individual rationality, budget balance, computational efficiency, and expected truthfulness at the same time. Our main contributions can be summarized as follows.
\begin{enumerate}
	\item We propose an novel edge-thing architecture based on blockchain technology to achieve complete decentralization and tempering-proof, in which the built-in smart contract acts as a central coordinator.
	\item To model a real edge-thing system, we formulate a combinatorial double auction model to achieve resources allocation between IoT devices and edge nodes.
	\item We introduce the exponential mechanism in differential privacy to our auction model so as to ensure privacy protection, and also achieve the expected truthfulness and approximately high revenue.
	\item We conduct extensive simulations to evaluate the performances of our proposed mechanisms. The simulation results verify our theoretical analysis.
\end{enumerate}

\textbf{Orgnizations: }In Sec. \uppercase\expandafter{\romannumeral2}, we survey the-state-of-art work. In Sec. \uppercase\expandafter{\romannumeral3}, we introduce the edge-thing system model and define our problem formally. In Sec. \uppercase\expandafter{\romannumeral4}, we introduce the differential privacy describe the mechasnism design in detial. In Sec. \uppercase\expandafter{\romannumeral5}, we give the proofs of related properties. Finally, we evaluate our mechanisms by numerical simulations in Sec.\uppercase\expandafter{\romannumeral6} and show the conclusions in Sec. \uppercase\expandafter{\romannumeral7}.

\section{Related Work}
In recent years, the related research on resources allocation has attract wide attention in academia. Auction theory has been used in a series of related areas, such as mobile crowdsensing \cite{yang2015incentive} \cite{guo2021reliable} and energy trading \cite{wang2012designing} \cite{yassine2019double}. In mobile edge computing environment, Sun \textit{et al.} \cite{sun2018double} proposed a double auction mechanism to allocate computing power between IoT devices and edge nodes, where IoT devices can purchase computing power from edge nodes. Habiba \textit{et al.} \cite{habiba2019reverse} put forward a reverse auction framework in mobile edge computing based on position, which aimed at maximizing the utility of edge servers. Peng \textit{et al.} \cite{peng2020multiattribute} designed a multiattribute-based double auction mechanism in vehicular edge computing, where the matching is determined by both price and non-price factors. However, a trusted auctioneer is essential to realize the resources allocation by auction mechanisms, especially for double auction. In P2P distributed edge network, there is no entity suitable acting as an auctioneer that can guarantee the security and reliability.

The emergence of blockchain technology has potentially solved this dilemma. It maintains a decentralized ledger, and can work as the auctioneer by combining smart contracts. Sun \textit{et al.} \cite{sun2020joint} revised their previous work in \cite{sun2018double} by introducing blockchain to achieve a trustworthy platform. Jiao \textit{et al.} \cite{jiao2019auction} proposed an auction-based market model for the allocation of computing resources between miners and edge servers. Ding \textit{et al.} \cite{ding2020incentive} \cite{ding2020pricing} attempted to build a secure blockchain-based IoT system by attracting more IoT devices to purchase computing power from edge servers and participate in the consensus process, where they adopted a multi-leader multi-follower Stackelberg game. Guo \textit{et al.} \cite{guo2020double} proposed a secure and efficient charging scheduling system based on DAG-blockchain and double auction mechanism. However, all the transaction models in these works are based on the allocation for a kind of resource. They did not consider the allocation of multiple resources. Moreover, they did not consider the potential risk of privacy disclosure.

Because of the public auction results, the sensitive information of participants is at risk of being exposed. To prevent the adversary from inferring players' sensitive information, Dwork \textit{et al.} \cite{dwork2008differential} founded the theory of differential privacy. McSherry \textit{et al.} \cite{mcsherry2007mechanism} first applied the differential privacy to auction mechanism and made a complete theoretical analysis. Chen \textit{et al.} \cite{chen2019differentially} combined the differential privacy with double spectrum auction design in order to maximize social welfare approximately. Guo \textit{et al.} \cite{guo2021differential} revised their work in \cite{guo2020double} by introducing differential privacy to avoid the leakage of bidding/asking information. Besides, differential privacy has been used in mechanism design of spectrum auction \cite{zhu2014differentially} \cite{zhu2015differentially}, smart grid \cite{xiang2016auc2reserve} \cite{li2019towards}, and mobile crowdsensing \cite{gao2019dpdt}. However, applying differential privacy to our combinatorial double auction model is very different from the existing work.

\section{Edge-Thing System Model}
In this section, we introduce the system model of the edge-thing architecture and how to integrate the blockchain as an effective technique to overcome the potential security threats in detail. Here, we consider the time can be discretized into time slots, denoted by $T=\{t_1,t_2,t_3,\cdots\}$, where each time slot is equal in length. The following discussion is within a time slot, including the combinatorial auction mechanism and consensus process. Finally, the objective function and problem definition can be formulated.
\subsection{System Description}
In the existing intelligent environment, there are a large number of IoT devices deployed in every corner of our lives, which undertake their own different tasks, such as traffic monitoring, health recording, navigation, and machine learning training. Because of their lightweight nature (limited resources) and delay sensitivity, these IoT devices can attempt to offload their tasks to adjacent edge nodes. In order to quantify the demand for different resources, we assume there are $k$ kinds of resources in our system, denoted by $\mathbb{R}=\{r_1,r_2,\cdots,r_k\}$, where each $r_i\in\mathbb{R}$ represents a certain kind of resource such as computation, memory, storage, or network bandwidth.

A certain number of edge nodes can form an intermediate layer between the more powerful cloud center and mobile IoT devices. In our system, there are $m$ IoT devices, denoted by $\mathbb{TD}=\{TD_1,\cdots,TD_i,\cdots,TD_m\}$. These IoT devices have limited computing power and storage space, thus not enough to achieve their goals. In order to upgrade the quality of service, the resources that are required by the IoT device $TD_i$ can be expressed as $\mathbb{D}_i=\{d_i^1,d_i^2,\cdots,d_i^k\}$, where $d_i^z\in[d_{min},d_{max}]$. Each $d_i^z\in\mathbb{D}_i$ indicates that IoT device $TD_i$ requires at least $d_i^z$ units of the resource $r_z$. Similarly, there are $n$ edge nodes, denoted by $\mathbb{EN}=\{EN_1,\cdots,EN_j,\cdots,EN_n\}$. These edge nodes are responsible for providing different resources to IoT devices. The resources that are provided by the edge node $EN_j$ can be expressed as $\mathbb{H}_j=\{h_j^1,h_j^2,\cdots,h_j^k\}$, where $h_j^z\in[h_{min},h_{max}]$. Each $h_j^z\in\mathbb{H}_j$ indicates that edge node $EN_j$ provides at most $h_j^z$ units of the resource $r_z$. Therefore, each resource-limited IoT device has to broadcast its resource request to the edge service provider in the hope of getting the resources it wants.

As mentioned earlier, in each time slot, edge nodes make a profit by selling resources and IoT devices complete their tasks by buying resources, which has created a double auction problem. In a double auction model, all players have to submit their requests to the auctioneer. There is an important question about who will assume the role of auctioneer. A natural idea is to let the cloud center be the auctioneer. However, this deviates from our original intention of getting rid of the cloud centers. There are several potential security threats when trying out a centralized cloud center, which can be summarized as follows.
\begin{enumerate}
	\item Vulnerability: the cloud center is attacked or damaged by malicious attackers or unexpected disasters. It will cause the single point of failure.
	\item Insecurity: the bidding/asking information submitted by players could be leaked or tampered with. It will cause data loss and privacy leakage.
	\item Unreliability: the cloud center is biased, and colludes with some nodes for their own benefits. It will cause the auction results to be unfair.
	\item Communication security and network delay: the cloud center is physically far away from IoT devices and edge nodes, which will cause potential security hazards and network delays in the transmission process.
\end{enumerate}

\begin{figure}[!t]
	\centering
	\includegraphics[width=\linewidth]{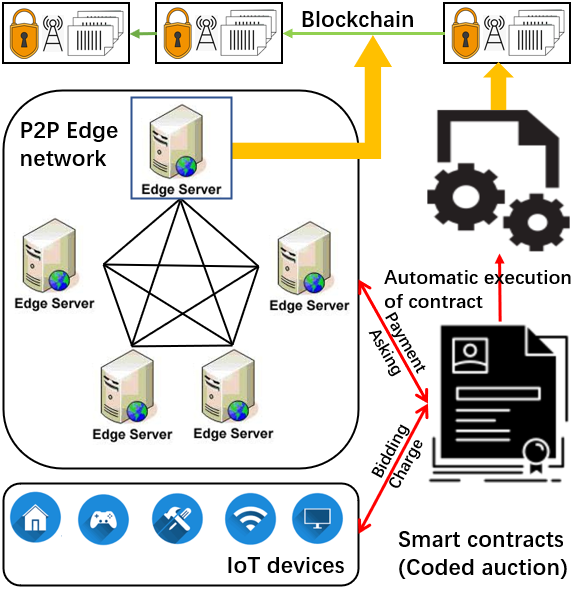}
	\centering
	\caption{The architecture of our edge-thing system based on blockchain and smart contract.}
	\label{fig1}
\end{figure}

In order to overcome the above drawbacks and achieve the decentralization, the blockchain and smart contract are used as ancillary techniques to prevent tampering and establish a credible system among unfamilar nodes without the third-party authority. The transaction between IoT devices and edge nodes are stored in the blockchain. Figure \ref{fig1} exhibits the architecture of our blockchain-enabled edge-thing system. Shown as Figure \ref{fig1}, IoT devices are light nodes that do not store the blockchain but participate in the transaction. Edge nodes are full nodes that store the complete blockchain and perform the consensus process to add new blocks to the blockchain. Moreover, a smart contract is deployed on the blockchain, which plays the role of auctioneer by implementing information interaction between IoT devices and edge nodes, and executing the predefined auction mechanism automatically. Such a system does not rely on a third-party authority to act the auctioneer, and also inherits the advantages of decentralization, temper resistance, and transparency in the blockchain.

\subsection{Combinatorial Auction Mechanism}
In order to simulate the real situation, we assume that each IoT device submitted its resource request in a bundled way, which formulates a combinatorial auction. For example, an IoT device needs to complete a task of training a deep learning model, thereby it wants to buy computation and memory from edge nodes. It is more reasonable to give a total bid according to its valuation of this task instead of bidding each resource separately. Furthermore, we find that this task can only be accomplished at one edge node. In other words, computing and memory resources must come from the same edge node, which increases the limitation of our model.

In a typical auction, there are three key roles, namely, the buyer, seller, and auctioneer. In our system, IoT devices are buyers, thus buyer set is $\mathbb{TD}$; edge nodes are sellers, thus seller set is $\mathbb{EN}$; and a smart contract is the auctioneer. In each time slot, the buyer requests a set of resources and gives the maximum price it is ready to pay to the edge node for buying these resources. For each buyer $TD_i\in\mathbb{TD}$, its bidding information can be denoted by $\mathcal{B}_i=(\mathbb{D}_i,b_i,dm_i)$ where the $b_i\in[v_{min},v_{max}]$ is the total bid (maximum buying price) to buy a bundle of resources $\mathbb{D}_i$, and the $dm_i$ is the maximum tolerant distance from the edge node providing resources to it. For each seller $EN_j\in\mathbb{EN}$, its asking information can be denoted by $\mathcal{A}_j=(\mathbb{H}_j,\vec{a}_j)$, where the $\vec{a}_j=(a_j^1,a_j^2,\cdots,a_j^k)$ is the asking vector where each $a_j^z\in\vec{a}_j$ is the unit ask (minimum selling price) per resource $r_z$. The bidding information of buyers and asking information of sellers are submitted to the auctioneer, therefore this auction can be defined as
\begin{equation}
	\Omega=\left(\left\{\mathcal{B}_i\right\}_{TD_i\in\mathbb{TD}},\left\{\mathcal{A}_j\right\}_{EN_j\in\mathbb{EN}}\right).
\end{equation}
Besides, for each buyer $TD_i\in\mathbb{TD}$, its valuation for obtaining the bundle of resources $\mathbb{D}_i$ is $v_i\in[v_{min},v_{max}]$. For each seller $EN_j\in\mathbb{EN}$, it has a cost vector $\vec{c}_j=(c_j^1,c_j^2,\cdots,c_j^k)$ where each $c_j^z\in[c_{min},c_{max}]$ is the unit cost per resource $r_z$.

In each time slot, once collecting the biding and asking information from players, the auctioneer will determine who are winning buyers and sellers, and how to allocate resources between them. The resource allocation is denoted by a binary matrix $\vec{X}_{m\times n}$, called ``allocation matrix''. For each $x_{ij}\in\vec{X}$, $x_{ij}=1$ if the resources requested by $TD_i$ are provided by $EN_j$ according to the result; otherwise $x_{ij}=0$. Besides, the auctioneer needs to determine the clearing price of each resource, which can be denoted by a price vector $\vec{p}=(p_1,p_2,\cdots,p_k)$. For each $p_z\in\vec{p}$, it is the unit price that buyers have to pay to get a unit of resources $r_z$.

\begin{rem}
	Here, we have $c_j^z\in[c_{min},c_{max}]$ for each resource $r_z\in\mathbb{R}$ and the price vector $\vec{p}\in[c_{min},c_{max}]^k$. For simplicity, we denoted by $\Theta=[c_{min},c_{max}]$ in the following description.
\end{rem}

\subsection{Problem Formulation}
According to the above definitions, we assume that the utility of each buyer $TD_i$ is denoted by $u_i^{TD}$ and the utility of each seller $EN_j$ is denoted by $u_j^{EN}$. After the auctioneer determines a clearing price vector $\vec{p}$ and its corresponding allocation matrix $\vec{X}$, the utilities of all losing players are equal to zero. Namely, we have $u_i^{TD}=0$ for each losing buyer $TD_i\in\mathbb{TD}$ if $\sum_{j=1}^{n}x_{ij}=0$ and $u_j^{EN}=0$ for each losing seller $EN_j\in\mathbb{EN}$ if $\sum_{i=1}^{m}x_{ij}=0$. The utility of each winning buyer is the difference between its valuation and payment toward its requested resources. In summary, for each buyer $TD_i\in\mathbb{TD}$, we have
\begin{equation}\label{eq2}
	u_i^{TD}=v_i-\sum_{j=1}^{n}x_{ij}\cdot\sum_{z=1}^{k}p_z\cdot d_i^z.
\end{equation}
The utility of each winning seller if the difference between the total payment from buyers and total cost. In summary, for each seller $EN_j\in\mathbb{EN}$, we have
\begin{equation}\label{eq3}
	u_j^{EN}=\sum_{z=1}^{k}(p_z-c_j^z)\cdot\sum_{i=1}^{m}x_{ij}\cdot d_i^z.
\end{equation}
Because the requested resources of a buyer must come from the same seller, there is a constraint that $\sum_{j=1}^{n}x_{ij}\leq1$ for each buyer $TD_i\in\mathbb{TD}$.

The result of an auction depends on its objective. In this system, we aim at maximizing the revenue of edge computing platform. The corresponding optimization problem can be summarized as to maximize the accumulated utility of all edge nodes, which is shown as the following problem:
\begin{align}
	\max\quad & \sum_{j=1}^{n}\left[\sum_{z=1}^{k}(p_z-a_j^z)\cdot\sum_{i=1}^{m}x_{ij}\cdot d_i^z\right]\label{eq4}\\
	s.\;t.\quad & \sum_{j=1}^{n}x_{ij}\leq1, \;\forall TD_i\in\mathbb{TD} \tag{\ref{eq4}{a}}\\
	& \sum_{i=1}^{m}x_{ij}\cdot d_i^z\leq h_j^z, \;\forall r_z\in\mathbb{R},\forall EN_j\in\mathbb{EN} \tag{\ref{eq4}{b}}\\
	& \sum_{j=1}^{n}x_{ij}\cdot \delta_{ij}\leq dm_i, \;\forall TD_i\in\mathbb{TD} \tag{\ref{eq4}{c}}\\
	& x_{ij}\in\{0,1\}, \;\forall TD_i\in\mathbb{TD},\forall EN_j\in\mathbb{EN}. \tag{\ref{eq4}{d}}
\end{align}
where $\delta_{ij}$ is the transmission distance between buyer $TD_i$ and seller $EN_j$. Constraint (4a) represents the many-to-one relationship from IoT devices to an edge node. Constraint (4b) states that the total consumption of each kind of resource $r_z$ cannot be larger than the maximum amount $h^z_j$ that can be provided by an edge node $EN_j$. Constraint (4c) implies that the distance between an IoT device $TD_i$ and the edge node that provides it with resources cannot be larger than the maximum distance allowed by this IoT device. This optimization problem can be classified as an integer linear programming problem, thus it is NP-hard.

\section{Mechanism Design}
In this section, we first introduce basic principles of designing an effective combinatorial auction mechanism. Due to the use of blockchain technology, the transactions in this system become transparent. That is to say, all auction results, including winners and clearing price, will be made public, which makes the system face the threat of inference attack. Therefore, we introduce the differential privacy into our mechanism design to avoid accidental disclosure of users' bidding/asking information.

\subsection{Design Rationales}
An effective auction mechanism has to satisfy the following four properties: individual rationality, budget balance, computational efficiency, and truthfulness.
\begin{defn}[Individual Rationality]
	An auction is individually rational if and only if the utilities of all players are non-negative. In our auction $\Omega$, we have $u_i^{TD}\geq 0$ for each buyer $TD_i\in\mathbb{TD}$ and $u_j^{EN}\geq 0$ for each seller $TN_j\in\mathbb{TN}$, where $u_i^{TD}$ and $u_j^{EN}$ are defined in (\ref{eq2}) and (\ref{eq3}).
\end{defn}
\begin{defn}[Budget Balance]
	An auction is budget balanced if and only if the auctioneer is profitable. In our auction $\Omega$, that is
	\begin{equation}
		\sum_{j=1}^{n}\left[\sum_{z=1}^{k}(p_z-a_j^z)\cdot\sum_{i=1}^{m}x_{ij}\cdot d_i^z\right]\geq 0.
	\end{equation}
\end{defn}
\begin{defn}[Computational Efficiency]
	The auction result can be obtained in polynomial time. 
\end{defn}
In an auction, players could manipulate their bids/asks in a strategical sense in order to win the auction. The truthfulness is a concept that encourages players in an auction to bid/ask according to their valuations/costs strictly. However in some cases, it is difficult to reach an exact truthfulness. Thus, we can consider an approximate truthfulness instead, called $\gamma$-truthfulness \cite{gupta2010differentially}, which ensures there is no one gaining more than $\gamma$ utility when bidding/asking truthfully.

\begin{defn}[$\gamma$-truthfulness]
	An auction is approximately truthful if and only if each player bids/asks truthfully is approximate to one of its dominant strategies. In our auction $\Omega$, for each buyer $TD_i\in\mathbb{TD}$, we have
	\begin{equation}
		\mathbb{E}\left[u_i^{TD}(v_i,\Omega_{-i})\right]\geq\mathbb{E}\left[u_i^{TD}(b_i,\Omega_{-i})\right]-\gamma
	\end{equation}
	where $\Omega_{-i}$ is other players' strategies except $TD_i$. For each seller $EN_j\in\mathbb{EN}$, we have
	\begin{equation}
		\mathbb{E}\left[u_j^{EN}(\vec{c}_j,\Omega_{-j})\right]\geq\mathbb{E}\left[u_j^{EN}(\vec{a}_j,\Omega_{-j})\right]-\gamma
	\end{equation}
	where $\Omega_{-j}$ is other players' strategies except $EN_j$.
\end{defn}

When we discuss the truthfulness in our auction, we assume that the resquested bundle $\mathbb{D}_i$ submitted by the buyer and the total resources $\mathbb{H}_j$ submitted by the seller are all believable because they can be monitored. Due to the truthfulness, no player is motivated to manipulate its strategy to gain more utility, which makes the strategic decision of players easier and guarantees a fair competitive environment.

\subsection{Differential Privacy}
The blockchain applied in our system can only ensure the security at the physical level, but it cannot prevent inference attacks. A curious player can infer other players' strategies by changing its own bid/ask in continuous auction rounds and analyzing the relavant auction results. With the help of other players' strategies, the attacker is able to make decisions in their favor and increase its benefits, thus undermining the fairness. To prevent this kind of threat, we choose to design a differentially private auction mechanism. Differential privacy is a technique that makes the attacker not distinguish between two neighboring inputs with high probability \cite{dwork2008differential}. Two datasets, $\vec{s}=(s_1,s_2,\cdots,s_i,\cdots)$ and $\vec{s}'=(s_1,s_2,\cdots,s_i',\cdots)$, are neighboring if and only if they have exactly one different element. For convenience, we denote by the bids of all buyers $\vec{b}=(b_1,b_2,\cdots,b_m)$ and the asks of all sellers $\vec{A}=(\vec{a}_1,\vec{a}_2,\cdots,\vec{a}_n)$. The definition of differential privacy is shown as follows.

\begin{defn}[Differential Privacy]\label{def5}
	We simplify our auction mechanism as a function $M(\cdot)$ that maps input bids $\vec{b}$ and input asks $\vec{A}$ to a clearing price $\vec{p}$. The mechanism $M(\cdot)$ gives $\varepsilon$-differential privacy if and only if, for any two neighboring inputs $(\vec{b},\vec{A})$ and $\{(\vec{b}',\vec{A})$ or $(\vec{b},\vec{A}')\}$, we have
	\begin{equation}
		\Pr[M(\vec{b},\vec{A})=\vec{p}]\leq\exp(\varepsilon)\cdot\Pr[(M(\vec{b}',\vec{A})=\vec{p}]
	\end{equation}
	\begin{equation}
		\Pr[M(\vec{b},\vec{A})=\vec{p}]\leq\exp(\varepsilon)\cdot\Pr[(M(\vec{b},\vec{A}')=\vec{p}]
	\end{equation}
	where the constant $\varepsilon$ is privacy budget.
\end{defn}

The privacy budget is a parameter for controlling the degree of privacy protection that a mechanism gives. Generally speaking, the smaller the privacy budget, the stronger the privacy protection. By introducing the differential privacy into our auction mechanism, the change of a player's bid/ask will not significantly affect the final clearing price. Thus, it prevents us from inference attacks through manipulating strategies and analyzing auction results.

Exponential mechanism \cite{dwork2008differential} is one of the mainstream methods to realize practical differential privacy. It depends on an ``score'' function $Q(\cdot)$ that maps input/output pairs to scores. The score function in our auction can be defined as $Q((\vec{b},\vec{A}),\vec{p})$, where a candidate output is more likely to be chosen if its score is higher. Thus, the exponential mechanism can be defined as follows.

\begin{defn}[Exponential Mechanism]
	Given an output $\vec{p}\in\Theta^k$, a score function $Q(\cdot)$, and a privacy budget $\varepsilon$, the exponential mechanism $M(\vec{b},\vec{A})$ selects $\vec{p}$ as its output with a probability that is proportional to its score $\varepsilon Q((\vec{b},\vec{A}),\vec{p})$. Thus, we have
	\begin{equation}
		\Pr[M(\vec{b},\vec{A})=\vec{p}]\propto\exp\left(\frac{\varepsilon Q((\vec{b},\vec{A}),\vec{p})}{2\Delta Q}\right)
	\end{equation}
	where $\Delta Q$ is the sensitivity of score function $Q(\cdot)$. That is the largest difference of their sorces for any two neighboring inputs $(\vec{b},\vec{A})$ and $\{(\vec{b}',\vec{A})$ or $(\vec{b},\vec{A}')\}$, which can be denoted by $\Delta Q=\max_{\vec{p}}\max_{(\vec{b},\vec{A}),(\vec{b}',\vec{A}')}\{|Q((\vec{b},\vec{A}),\vec{p})-Q((\vec{b}',\vec{A}),\vec{p})|,|Q((\vec{b},\vec{A}),\vec{p})-Q((\vec{b},\vec{A}'),\vec{p})|\}$.
\end{defn}

\subsection{Algorithm Design and Description}
The design goal of our auction mechanism is to maximize the revenue of edge computing platform approximately, but achieve $\varepsilon$-differential privacy, $\gamma$-truthfulness, individual rationality, budget balance, and computational efficiency at the same time. The mechanism can be divided into three stages, winning candidate determination, assigment, and pricing. The procedure is shown in Algorithm \ref{alg1}.

\begin{algorithm}[!t]
	\caption{\text{DPAM}}\label{alg1}
	\begin{algorithmic}[1]
		\renewcommand{\algorithmicrequire}{\textbf{Input:}}
		\renewcommand{\algorithmicensure}{\textbf{Output:}}
		\REQUIRE $(\{\mathcal{B}_i\}_{TD_i\in\mathbb{TD}},\{\mathcal{A}_j\}_{EN_j\in\mathbb{EN}})$, $\varepsilon$, $\Theta$
		\ENSURE $\vec{X}_{\vec{p}}$, $\vec{p}$
		\STATE Initialize $\Delta R=\sum_{j=1}^{n}(c_{max}-c_{min})\cdot\sum_{z=1}^{k}h_j^z$
		\FOR {each $\vec{p}\in\Theta^k$}
		\STATE \textit{// Winning candidate determination}
		\STATE Initialize $x_{ij}=0$ for each $x_{ij}\in\vec{X}_{\vec{p}}$
		\STATE Initialize $\mathbb{TD}_c\leftarrow\emptyset$
		\FOR {each $TD_i\in\mathbb{TD}$}
		\IF {$\sum_{z=1}^{k}p_z\cdot d_i^z\leq b_i$}
		\STATE $\mathbb{TD}_c\leftarrow\mathbb{TD}_c\cup\{TD_i\}$
		\ENDIF
		\ENDFOR
		\STATE Sort the $\mathbb{TD}_c$ s.t. $\sum_{z=1}^{k}d_1^z\geq\sum_{z=1}^{k}d_2^z\geq\cdots$
		\STATE \textit{// Assignment}
		\STATE Initialize $\{{h_j^1}',{h_j^2}',\cdots,{h_j^k}'\}$ where ${h_j^z}'=h_j^z\in\mathbb{H}_j$
		\FOR {each $TD_i\in\mathbb{TD}_c$}
		\STATE Initialize $\mathbb{EN}_{c,i}\leftarrow\emptyset$
		\FOR {each $EN_j\in\mathbb{EN}$}
		\IF {${h_j^z}'\geq d_i^z$ for each $r_z\in\mathbb{R}$, $\delta_{ij}\leq dm_i$, and $\sum_{z=1}^{k}(p_z-a_j^z)\cdot d_i^z\geq 0$}
		\STATE $\mathbb{EN}_{c,i}\leftarrow\mathbb{EN}_{c,i}\cup\{EN_j\}$
		\ENDIF
		\ENDFOR
		\IF {$\mathbb{EN}_{c,i}\neq\emptyset$}
		\STATE $EN_{j^*}\leftarrow\arg\min_{EN_j\in\mathbb{EN}_{c,i}}\{\delta_{ij}\}$
		\FOR {each $r_z\in\mathbb{R}$}
		\STATE ${h_{j^*}^z}'\leftarrow{h_{j^*}^z}'-d_i^z$
		\ENDFOR
		\STATE $x_{ij^*}\leftarrow 1$
		\ENDIF
		\ENDFOR
		\STATE $R((\vec{b},\vec{A}),\vec{p})=\sum_{j=1}^{n}[\sum_{z=1}^{k}(p_z-a_j^z)\sum_{i=1}^{m}x_{ij} d_i^z]$
		\ENDFOR
		\STATE \textit{// Pricing}
		\STATE Select a $\vec{p}\in\Theta^k$ according to the selection distribution: 	$\Pr[M(\vec{b},\vec{A})=\vec{p}]=\frac{\exp\left(\frac{\varepsilon R((\vec{b},\vec{A}),\vec{p})}{2\Delta R}\right)}{\sum_{\vec{p}'\in\Theta^k}\exp\left(\frac{\varepsilon R((\vec{b},\vec{A}),\vec{p}')}{2\Delta R}\right)}$
		\RETURN $\vec{X}_{\vec{p}}$, $\vec{p}$
	\end{algorithmic}
\end{algorithm}

In the winning candidate determination, we first select a subset of $\mathbb{TD}$ as winning buyer candidates, which is denoted by $\mathbb{TD}_c\subseteq\mathbb{TD}$. Given a price vector $\vec{p}\in\Theta^k$, we have $TD_i\in\mathbb{TD}_c$ if and only if it satisfies
\begin{equation}
	\sum_{z=1}^{k}p_z\cdot d_i^z\leq b_i.
\end{equation}
Then, we sort the set of winning buyer candidates $\mathbb{TD}_c$ in a descending order according to their amount of requested resources. For each buyer $TD_i\in\mathbb{TD}_c$, its amount of requested resources is defined as $\sum_{z=1}^{k}d_i^z$. Thus, we sort $\mathbb{TD}_c=\{TD_1,TD_2,\cdots\}$ where they satisfy $\sum_{z=1}^{k}d_1^z\geq\sum_{z=1}^{k}d_2^z\geq\cdots$ definitely. Next, for each buyer $TD_i\in\mathbb{TD}_c$, we need to determine its winning seller candidates, which is denoted by $\mathbb{EN}_{c,i}\in\mathbb{EN}$. We have $EN_j\in\mathbb{EN}_{c,i}$ if and only if it satisfies three conditions.
\begin{enumerate}
	\item Its remaining resources $\mathbb{H}_j'$ are sufficient. In other words, we have ${h_j^z}'\geq d_i^z$ for each $r_z\in\mathbb{R}$.
	\item Its distance from $TD_i$ is close enoguh. Thus, we have $\delta_{ij}\leq dm_i$.
	\item It is profitable by providing resources to buyer $TD_i$. Here, we have $\sum_{z=1}^{k}(p_z-a_z)\cdot d_i^z\geq 0$.
\end{enumerate}
Conditon (1) and (2) is obvious. If Condition (3) cannot be satisfied, providing resources for $TD_i$ by $EN_j$ ($x_{ij}=1$) will lead to a decrease in the objective value.

Given a buyer $TD_i\in\mathbb{TD}_c$, we can get its winning seller candidates $\mathbb{EN}_{c,i}$. If $|\mathbb{EN}_{c,i}|\geq 1$, how can we select the best one to provide resources? In the assignment stage, we can think about it in two directions. The first strategy is to consider load balancing, and we try our best to arrange the edge node with more idle resources to provide service. The second strategy is to consider saving network bandwidth, and we try our best to arrange the edge node that is closest to the target buyer $TD_i$. Because an edge node can provide a variety of resources, how to quantify "idle resources" is difficult. Thus, we use the second strategy here, where we select an $EN_{j^*}\in\mathbb{EN}_{c,i}$ that satisfies
\begin{equation}
	EN_{j^*}=\arg\min_{EN_j\in\mathbb{EN}_{c,i}}\{\delta_{ij}\}
\end{equation}
to provide resources to the buyer $TD_i$.

From the above process, we can obtain winning buyers and winning sellers, and their corresponding allocation matrix $\vec{X}$ given a price vector $\vec{p}$. The next pricing stage is to determine which price vector $\vec{p}\in\Theta^k$ we select. This pricing process comes from both the uniform pricing \cite{son2004short} and the exponential mechanism. Given a price vector $\vec{p}$, it generates an allocation matrix $\vec{X}_{\vec{p}}$, then we can calculate the corresponding revenue of the platform as (\ref{eq4}), denoted by
\begin{equation}\label{eq13}
	R((\vec{b},\vec{A}),\vec{p})=\sum_{j=1}^{n}\left[\sum_{z=1}^{k}(p_z-a_j^z)\cdot\sum_{i=1}^{m}x_{ij}\cdot d_i^z\right].
\end{equation}
We make this platform revenue as the score of price $\vec{p}$. The sensitivity of score function $R(\cdot)$ can be formulated as
\begin{equation}\label{eq14}
	\Delta R=\sum_{j=1}^{n}(c_{max}-c_{min})\cdot\sum_{z=1}^{k}h_j^z
\end{equation}
since $p_z-a_j^z\leq c_{max}-c_{min}$ and $\sum_{i=1}^{m}x_{ij}\cdot d_i^z\leq h_j^z$. Then, we repeat the above process to calculate platform revenues under all possible price $\vec{p}\in\Theta^k$. To determine the final pricing, we define the probability distribution of price vectors as follows.
\begin{equation}\label{eq15}
	\Pr[M(\vec{b},\vec{A})=\vec{p}]=\frac{\exp\left(\frac{\varepsilon R((\vec{b},\vec{A}),\vec{p})}{2\Delta R}\right)}{\sum_{\vec{p}'\in\Theta^k}\exp\left(\frac{\varepsilon R((\vec{b},\vec{A}),\vec{p}')}{2\Delta R}\right)}
\end{equation}
where $R((\vec{b},\vec{A}),\vec{p})$ is defined in (\ref{eq13}) and $\Delta R$ is defined in (\ref{eq14}). Given all possible price $\vec{p}\in\Theta^k$ and their scores, in the pricing stage, it randomly select a price vector $\vec{p}$ with the probability $\Pr[M(\vec{b},\vec{A})=\vec{p}]$ shown as (\ref{eq15}).

\section{Theoretical Analysis}
In this section, we describe the theoretical analysis of how our proposed mechanism DPAM, shown as Algorithm \ref{alg1} satisfies desirable properties.
\begin{thm}\label{thm1}
	The DPAM achieves $\varepsilon$-differential privacy.
\end{thm}
\begin{proof}
	Given two neighboring inputs $(\vec{b},\vec{A})$ and $(\vec{b}',\vec{A})$, the mechanism $M$ randomly select a clearing price $\vec{p}$ from $\Theta^k$. Thus, the probability ratio of their corresponding probability selected by the $M$ is shown as follows.
	\begin{align}
		&\frac{\Pr[M(\vec{b},\vec{A})=\vec{p}]}{\Pr[M(\vec{b}',\vec{A})=\vec{p}]}\nonumber\\
		&=\frac{\exp\left(\frac{\varepsilon R((\vec{b},\vec{A}),\vec{p})}{2\Delta R}\right)}{\sum_{\vec{p}'\in\Theta^k}\exp\left(\frac{\varepsilon R((\vec{b},\vec{A}),\vec{p}')}{2\Delta R}\right)}/\frac{\exp\left(\frac{\varepsilon R((\vec{b}',\vec{A}),\vec{p})}{2\Delta R}\right)}{\sum_{\vec{p}'\in\Theta^k}\exp\left(\frac{\varepsilon R((\vec{b}',\vec{A}),\vec{p}')}{2\Delta R}\right)}\nonumber\\
		&=\exp\left(\frac{\varepsilon[R((\vec{b},\vec{A}),\vec{p})-R((\vec{b},\vec{A}),\vec{p})]}{2\Delta R}\right)\cdot\nonumber\\
		&\quad\frac{\sum_{\vec{p}'\in\Theta^k}\exp\left(\frac{\varepsilon R((\vec{b}',\vec{A}),\vec{p}')}{2\Delta R}\right)}{\sum_{\vec{p}'\in\Theta^k}\exp\left(\frac{\varepsilon R((\vec{b},\vec{A}),\vec{p}')}{2\Delta R}\right)}\nonumber\\
		&\leq\exp\left(\frac{\varepsilon}{2}\right)\cdot\frac{\sum_{\vec{p}'\in\Theta^k}\exp\left(\frac{\varepsilon[R((\vec{b},\vec{A}),\vec{p}')+\Delta R]}{2\Delta R}\right)}{\sum_{\vec{p}'\in\Theta^k}\exp\left(\frac{\varepsilon R((\vec{b},\vec{A}),\vec{p}')}{2\Delta R}\right)}\nonumber\\
		&\leq\exp\left(\frac{\varepsilon}{2}\right)\cdot\exp\left(\frac{\varepsilon}{2}\right)\cdot\frac{\sum_{\vec{p}'\in\Theta^k}\exp\left(\frac{\varepsilon R((\vec{b},\vec{A}),\vec{p}')}{2\Delta R}\right)}{\sum_{\vec{p}'\in\Theta^k}\exp\left(\frac{\varepsilon R((\vec{b},\vec{A}),\vec{p}')}{2\Delta R}\right)}\nonumber\\
		&=\exp(\varepsilon).\nonumber
	\end{align}
	By symmetry, we have $\Pr[M(\vec{b},\vec{A})=\vec{p}]/\Pr[M(\vec{b}',\vec{A})=\vec{p}]\geq\exp(-\varepsilon)$. According to Definition \ref{def5}, the DPAM is $\varepsilon$-differentially private to buyers.
	
	Given two neighboring inputs $(\vec{b},\vec{A})$ and $(\vec{b},\vec{A}')$, by similar induction procedure as buyers, we have
	\begin{equation*}
		\Pr[M(\vec{b},\vec{A})=\vec{p}]/\Pr[M(\vec{b},\vec{A}')=\vec{p}]\leq\exp(\varepsilon).
	\end{equation*}
	Thus, the DPAM is $\varepsilon$-differentially private to sellers, and Theorem \ref{thm1} has been proven.
\end{proof}

To achieve the $\gamma$-truthfulness eventually, we first introduce the following two lemmas as a foreshadowing.
\begin{lem}\label{lem1}
	Given a clearing price $\vec{p}\in\Theta^k$, for each buyer $TD_i\in\mathbb{TD}$, the DPAM achieves
	\begin{equation}
		u_i^{TD}((v_i,\Omega_{-i}),\vec{p})\geq u_i^{TD}((b_i,\Omega_{-i}),\vec{p}).
	\end{equation}
\end{lem}
\begin{proof}
	The $TD_i\in\mathbb{TD}_c$ if it bids truthfully. There are two sub-cases we need to concern:
	\begin{itemize}
		\item $b_i>v_i$: The $TD_i$ will be in $\mathbb{TD}_c$ as well. According to the winner candidate determination and assignment, the auction result to the $TD_i$ will not change. Thus, we have $u_i^{TD}((v_i,\Omega_{-i}),\vec{p})=u_i^{TD}((b_i,\Omega_{-i}),\vec{p})$.
		\item $b_i<v_i$: If $\sum_{z=1}^{k}p_z\cdot d_i^z\leq b_i$ can be satisfied, the $TD_i$ will be in $\mathbb{TD}_c$ as well. Thus, we have $u_i^{TD}((v_i,\Omega_{-i}),\vec{p})=u_i^{TD}((b_i,\Omega_{-i}),\vec{p})$; Otherwise, the $TD_i$ will be not in $\mathbb{TD}_c$, which loses the auction definitely. Thus, we have $u_i^{TD}((v_i,\Omega_{-i}),\vec{p})\geq u_i^{TD}((b_i,\Omega_{-i}),\vec{p})=0$.
	\end{itemize}
	
	The $TD_i\notin\mathbb{TD}_c$ if it bids truthfully. There are two sub-cases we need to concern:
	\begin{itemize}
		\item $b_i>v_i$: If $\sum_{z=1}^{k}p_z\cdot d_i^z\leq b_i$ can be satisfied, the $TD_i$ will be in $\mathbb{TD}_c$. If it can be assigned an edge node in the assignment stage, its utilty will be $u_i^{TD}((b_i,\Omega_{-i}),\vec{p})=v_i-\sum_{z=1}^{k}p_z\cdot d_i^z<0=u_i^{TD}((v_i,\Omega_{-i}),\vec{p})$.
		\item $b_i<v_i$: The $TD_i$ will be not in $\mathbb{TD}_c$ as well, which loses the auction definitely. Thus, we have $u_i^{TD}((v_i,\Omega_{-i}),\vec{p})=u_i^{TD}((b_i,\Omega_{-i}),\vec{p})=0$.
	\end{itemize}

	From the above, we always have $u_i^{TD}((v_i,\Omega_{-i}),\vec{p})\geq u_i^{TD}((b_i,\Omega_{-i}),\vec{p})$, and Lemma \ref{lem1} has been proven.
\end{proof}

\begin{lem}\label{lem2}
	Given a clearing price $\vec{p}\in\Theta^k$, for each buyer $EN_j\in\mathbb{EN}$, the DPAM achieves
	\begin{equation}
		u_j^{EN}((\vec{c}_j,\Omega_{-j}),\vec{p})\geq u_j^{EN}((\vec{a}_j,\Omega_{-j}),\vec{p}).
	\end{equation}
\end{lem}
\begin{proof}
	First, ``$\vec{a}_j>\vec{c}_j$'' implies there is at least one element in these vectors satisfying $a_j^{z^*}>c_j^{z^*}$ and others satisfy $a_j^z\geq c_j^z$ for each $r_z\in\mathbb{R}\backslash\{r_{z^*}\}$. Second, we denoted by $x_{ij}\in\vec{X}$ the allocation when a seller asks truthfully and $\bar{x}_{ij}\in\bar{\vec{X}}$ the allocation when a seller asks untruthfully. 
	
	Consider the seller $EN_j\in\mathbb{EN}$, there are two sub-cases we need to concern:
	\begin{itemize}
		\item $\vec{a}_j>\vec{c}_j$: When $x_{ij}=1$, the auction result will be $\bar{x}_{ij}=1$ as well if $\sum_{z=1}^{k}(p_z-a_z)\cdot d_{i}^z\geq 0$ can be satisfied; otherwise $\bar{x}_{ij}=0$. Thus, we have $u_j^{EN}((\vec{c}_j,\Omega_{-j}),\vec{p})\geq u_j^{EN}((\vec{a}_j,\Omega_{-j}),\vec{p})$ because $x_{ij}\geq\bar{x}_{ij}$.	
		\item $\vec{a}_j<\vec{c}_j$: When $x_{ij}=1$, the auction result will be $\bar{x}_{ij}=1$ as well. When $x_{ij}=0$ and $\sum_{z=1}^{k}(p_z-c_z)\cdot d_i^z\geq 0$, the auction result to the $\bar{x}_{ij}=0$ will not change according to the assignment. However, when $x_{ij}=0$ and $\sum_{z=1}^{k}(p_z-c_z)\cdot d_i^z< 0$, it is possible to happen $\sum_{z=1}^{k}(p_z-c_z)\cdot d_i^z\geq0$, and leading to $\bar{x}_{ij}=1$ if $\delta_{ij}$ is the minimum one among this buyer's winning seller candidates. The utility gained from $TD_i$ less than zero. Thus, we have $u_j^{EN}((\vec{c}_j,\Omega_{-j}),\vec{p})\geq u_j^{EN}((\vec{a}_j,\Omega_{-j}),\vec{p})$.
	\end{itemize}

	From the above, we always have $u_j^{EN}((\vec{c}_j,\Omega_{-j}),\vec{p})\geq u_j^{EN}((\vec{a}_j,\Omega_{-j}),\vec{p})$, and Lemma \ref{lem2} has been proven.
\end{proof}

\begin{thm}\label{thm2}
	The DPAM achieves $\gamma$-truthfulness.
\end{thm}
\begin{proof}
    Given two neighboring inputs $(\vec{b},\vec{A})$ and $(\vec{b}',\vec{A})$, for any buyer $TD_i\in\mathbb{TD}$, we assume that $v_i\in\vec{b}$ and $b_i\in\vec{b}'$. Thus, we have
    \begin{align}
    	&\mathbb{E}\left[u_i^{TD}(v_i,\Omega_{-i})\right]\nonumber\\
    	&=\sum_{\vec{p}\in\Theta^k}\Pr[M(\vec{b},\vec{A})=\vec{p}]\cdot u_i^{TD}((v_i,\Omega_{-i}),\vec{p})\nonumber\\
    	&\geq\exp(-\varepsilon)\cdot\sum_{\vec{p}\in\Theta^k}\Pr[M(\vec{b}',\vec{A})=\vec{p}]\cdot u_i^{TD}((b_i,\Omega_{-i}),\vec{p})\nonumber\\
    	&=\exp(-\varepsilon)\cdot\mathbb{E}\left[u_i^{TD}(b_i,\Omega_{-i})\right]\nonumber\\
    	&\geq(1-\varepsilon)\cdot\mathbb{E}\left[u_i^{TD}(b_i,\Omega_{-i})\right]\nonumber\\
    	&\geq\mathbb{E}\left[u_i^{TD}(b_i,\Omega_{-i})\right]-\nonumber\varepsilon\cdot v_{max}.
    \end{align}
	For any buyer $TD_i\in\mathbb{TD}$, we have $\mathbb{E}[u_i^{TD}(b_i,\Omega_{-i})]\leq\max_{TD_i\in\mathbb{TD}}\{u^{TD}_i\}\leq v_{max}-c_{min}\cdot\min_{TD_i\in\mathbb{TD}}\{\sum_{z=1}^{k}d_i^z\}\leq v_{max}$. Thus, we can conclude that the DPAM achieves $\varepsilon\cdot v_{max}$-truthfulness to buyers.
	
	Given two neighboring inputs $(\vec{b},\vec{A})$ and $(\vec{b},\vec{A}')$, for any seller $EN_j\in\mathbb{EN}$, we assume that $\vec{c}_j\in\vec{A}$ and $\vec{a}_j\in\vec{A}'$. Similarly as the above, we have
	\begin{align}
		&\mathbb{E}\left[u_j^{EN}(\vec{c}_j,\Omega_{-j})\right]\nonumber\\
		&\geq\mathbb{E}\left[u_j^{EN}(\vec{a}_j,\Omega_{-j})\right]-\varepsilon\cdot(c_{max}-c_{min})\cdot k\cdot h_{max}.\nonumber
	\end{align}
	For any seller $EN_j\in\mathbb{EN}$, we have $\mathbb{E}[u_j^{EN}(\vec{c}_j,\Omega_{-j})]\leq\max_{EN_j\in\mathbb{EN}}\{u_j^{EN}\}\leq(c_{max}-c_{min})\cdot\sum_{z=1}^{k}h_j^z\leq(c_{max}-c_{min})\cdot k\cdot h_{max}$. Thus, we can conclude that the DPAM achieves $\varepsilon\cdot(c_{max}-c_{min})\cdot k\cdot h_{max}$-truthfulness to sellers.
	
	Giving $\gamma=\max\{\varepsilon\cdot v_{max}, \varepsilon\cdot(c_{max}-c_{min})\cdot k\cdot h_{max}\}$, the DPAM achieve $\gamma$-truthfulness.
\end{proof}

\begin{thm}\label{thm3}
	The DPAM achieves individual rationality.
\end{thm}
\begin{proof}
	According to Theorem \ref{thm2}, no player has the motivation to bid/ask untruthfully. We can consider $b_i=v_i$ for each buyer $TD_i\in\mathbb{TD}$ and $\vec{a}_j=\vec{c}_j$ for each seller $EN_j\in\mathbb{EN}$. Based on the winning candidate determination, each winning buyer $TD_i$ must have $\sum_{z=1}^{k}p_z\cdot d_i^z\leq b_i$, thus its utility $u^{TD}_i\geq 0$. Based on the assignment, each winning seller $EN_j$ that provides resources to buyer $TD_i$ must have $\sum_{z=1}^{k}(p_z-a_z)\cdot d_i^z\geq 0$, which means that providing resources to an IoT devices always bring positive returns. Thus, its utility $u^{EN}_j\geq 0$.
\end{proof}
\begin{thm}
	The DPAM achieves budget balanced.
\end{thm}
\begin{proof}
	The utilities of all edge nodes are positive according to Theorem \ref{thm3}, thus the sum of them (budget) $\sum_{j=1}^{n}u^{EN}_j$ is greater than zero as well.
\end{proof}
\begin{thm}\label{thm5}
	The DPAM does not achieve computational efficiency.
\end{thm}
\begin{proof}
	The main loop to traverse all possible price vectors $\vec{p}\in\Theta^k$ contains $|\Theta^k|$ iterations. For each iteration, the dominant step in winning candidate determination is to sort $\mathbb{TD}_c$, which has at most $m$ elements. Thus, sorting $\mathbb{TD}_c$ is bounded by $O(m\log m)$. Then, in the assignment, it takes $O(n)$ for each buyer $TD_i\in\mathbb{TD}_c$. Thus, its running time is bounded by $O(mn)$. The total time complexity of Algorithm \ref{alg1} is bounded by $O((mn+m\log m)\cdot|\Theta^k|)$. Therefore, the running time increases exponentially with $k$ instead of polynomial time.
\end{proof}

Next, we need to calculate the expected performance of our proposed mechanism. Based on (\ref{eq13}), the expected revenue of edge computing platform can be expressed as
\begin{equation}\label{eq18}
	\mathbb{E}[R(\vec{b},\vec{A})]=\sum_{\vec{p}\in\Theta^k}\Pr[M(\vec{b},\vec{A})=\vec{p}]\cdot R((\vec{b},\vec{A}),\vec{p}).
\end{equation}
To achieve the approximation ratio of the DPAM, we first introduce the following lemma.

\begin{lem}\label{lem3}
	Let $OPT$ be the optimal revenue by solving the problem defined in (\ref{eq4}) and $OPT^*=\max_{\vec{p}\in\Theta^k}\{R((\vec{b},\vec{A}),\vec{p})\}$ be the maximum revenue obtained by the winning candidate determination and assignment process of Algorithm \ref{alg1}. Then, we have
	\begin{equation}
		F(\Theta)\cdot OPT\leq OPT^*\leq OPT
	\end{equation}
	where we denoted by $F(\Theta)=\frac{\max_{\vec{p}\in\Theta^k}\{R((\vec{b},\vec{A}),\vec{p})\}}{(c_{max}-c_{min})\cdot n\cdot k\cdot h_{max}}$ as a factor of $OPT$.
\end{lem}
\begin{proof}
	Because the $OPT$ is globally optimal, we must have $OPT\geq OPT^*$. Based on (\ref{eq4}), we have
	\begin{equation}\label{eq20}
		OPT\leq(c_{max}-c_{min})\cdot n\cdot k\cdot h_{max}
	\end{equation}
	since each edge node provides at most $k\cdot h_{max}$ units of resources and there are total $n$ edge nodes. According to the definition of $OPT^*$, we have
	\begin{align}
		OPT^*&=\max_{\vec{p}\in\Theta^k}\{R((\vec{b},\vec{A}),\vec{p})\}\nonumber\\
		&\geq\frac{\max_{\vec{p}\in\Theta^k}\{R((\vec{b},\vec{A}),\vec{p})\}}{(c_{max}-c_{min})\cdot n\cdot k\cdot h_{max}}\cdot OPT\nonumber\\
		&=F(\Theta)\cdot OPT\nonumber
	\end{align}
	since the relationship $(\ref{eq20})$ exists.
\end{proof}
In order to achieve the truthfulness, the returned revenue is not optimal even though there is no differential privacy. This difference can be bounded by $F(\Theta)$. After introducing the differential privacy, the revenue will be damaged further. 

\begin{thm}\label{thm6}
	The expected revenue of edge computing platform $\mathbb{E}[R(\vec{b},\vec{A})]$ achieved by DPAM and the optimal revenue $OPT$ satisfies that $\mathbb{E}[R(\vec{b},\vec{A})]\geq$
	\begin{equation}
		F(\Theta)\cdot OPT-\frac{6\Delta R}{\varepsilon}\cdot\ln\left(e+\frac{\varepsilon OPT|\Theta^k|}{2\Delta R}\right).
	\end{equation}
\end{thm}
\begin{proof}
	Let $OPT^*=\max_{\vec{p}\in\Theta^k}\{R((\vec{b},\vec{A}),\vec{p})\}$ be the maximum revenue returned by the DPAM. For a small constant $t\geq 0$, we define four sets, which are $S_t=\{\vec{p}:R((\vec{b},\vec{A}),\vec{p})>OPT^*-t\}$, $\bar{S}_t=\{\vec{p}:R((\vec{b},\vec{A}),\vec{p})\leq OPT^*-t\}$,
	$S_{2t}=\{\vec{p}:R((\vec{b},\vec{A}),\vec{p})>OPT^*-2t\}$, and $\bar{S}_{2t}=\{\vec{p}:R((\vec{b},\vec{A}),\vec{p})\leq OPT^*-2t\}$. Thus, we have $\Pr[M(\vec{b},\vec{A})\in\bar{S}_{2t}]\leq$
	\begin{align}
		&\leq\frac{\Pr[M(\vec{b},\vec{A})\in\bar{S}_{2t}]}{\Pr[M(\vec{b},\vec{A})\in S_t]}\nonumber\\
		&=\frac{\sum_{\vec{p}\in\bar{S}_{2t}}\frac{\exp\left(\frac{\varepsilon R((\vec{b},\vec{A}),\vec{p})}{2\Delta R}\right)}{\sum_{\vec{p}'\in\Theta^k}\exp\left(\frac{\varepsilon R((\vec{b},\vec{A}),\vec{p}')}{2\Delta R}\right)}}{\sum_{\vec{p}\in S_t}\frac{\exp\left(\frac{\varepsilon R((\vec{b},\vec{A}),\vec{p})}{2\Delta R}\right)}{\sum_{\vec{p}'\in\Theta^k}\exp\left(\frac{\varepsilon R((\vec{b},\vec{A}),\vec{p}')}{2\Delta R}\right)}}\nonumber\\
		&=\frac{\sum_{\vec{p}\in\bar{S}_{2t}}\exp\left(\frac{\varepsilon R((\vec{b},\vec{A}),\vec{p})}{2\Delta R}\right)}{\sum_{\vec{p}\in S_t}\exp\left(\frac{\varepsilon R((\vec{b},\vec{A}),\vec{p})}{2\Delta R}\right)}\leq\frac{|\bar{S}_{2t}|\cdot\exp\left(\frac{\varepsilon(OPT^*-2t)}{2\Delta R}\right)}{|S_t|\cdot\exp\left(\frac{\varepsilon(OPT^*-t)}{2\Delta R}\right)}\nonumber\\
		&=\frac{|\bar{S}_{2t}|}{|S_t|}\cdot\exp\left(\frac{-\varepsilon t}{2\Delta R}\right)\label{eq22}.
	\end{align}
	Based on (\ref{eq22}), we have
	\begin{align}
		\Pr[M(\vec{b},\vec{A})\in S_{2t}]&\geq 1-\frac{|\bar{S}_{2t}|}{|S_t|}\cdot\exp\left(\frac{-\varepsilon t}{2\Delta R}\right)\nonumber\\
		&\geq 1-|\Theta^k|\cdot\exp\left(\frac{-\varepsilon t}{2\Delta R}\right)
	\end{align}
	since $\Pr[M(\vec{b},\vec{A})\in S_{2t}]+\Pr[M(\vec{b},\vec{A})\in\bar{S}_{2t}]=1$, $|\bar{S}_{2t}|\leq|\Theta^k|$, and $|S_t|\geq 1$. Thus, the expected revenue $\mathbb{E}[R(\vec{b},\vec{A})]$ can be expressed as
	\begin{align}
		\mathbb{E}[R(\vec{b},\vec{A})]&\geq\sum_{\vec{p}\in S_{2t}}\Pr[M(\vec{b},\vec{A})=\vec{p}]\cdot R((\vec{b},\vec{A}),\vec{p})\nonumber\\
		&\geq\Pr[M(\vec{b},\vec{A})\in S_{2t}]\cdot(OPT^*-2t)\nonumber\\
		&\geq\left[1-|\Theta^k|\cdot\exp\left(\frac{-\varepsilon t}{2\Delta R}\right)\right]\cdot(OPT^*-2t)\nonumber
	\end{align}
	For any $t$ satisfying 
	\begin{equation}\label{eq24}
		t\geq\frac{2\Delta R}{\varepsilon}\cdot\ln\left(\frac{|\Theta^k|OPT^*}{t}\right)
	\end{equation}
    we have $\exp\left(\frac{-\varepsilon t}{2\Delta R}\right)\leq\frac{t}{OPT^*|\Theta^k|}$. Thus,
	\begin{align}
		\mathbb{E}[R(\vec{b},\vec{A})]&\geq\left(1-|\Theta^k|\cdot\frac{t}{OPT^*|\Theta^k|}\right)\cdot(OPT^*-2t)\nonumber\\
		&=OPT^*-3t+\frac{2t^2}{OPT^*}\nonumber\\
		&\geq OPT^*-3t\label{eq25}.
	\end{align}
	By giving $t=\frac{2\Delta R}{\varepsilon}\ln\left(e+\frac{\varepsilon OPT^*|\Theta^k|}{2\Delta R}\right)$, we have
	\begin{align}
		t&=\frac{2\Delta R}{\varepsilon}\cdot\ln\left(e+\frac{\varepsilon OPT^*|\Theta^k|}{2\Delta R}\right)\nonumber\\
		&\geq\frac{2\Delta R}{\varepsilon}\cdot\ln\left(OPT^*|\Theta^k|\frac{\varepsilon}{2\Delta R}\right)\nonumber\\
		&\geq\frac{2\Delta R}{\varepsilon}\cdot\ln\left(\frac{|\Theta^k|OPT^*}{t}\right)\nonumber
	\end{align}
	where it satisfies (\ref{eq24}) because $\ln\left(e+\frac{\varepsilon OPT^*|\Theta^k|}{2\Delta R}\right)\geq1$ and $t\geq\frac{2\Delta R}{\varepsilon}$. Finally, we substitute $t=\frac{2\Delta R}{\varepsilon}\cdot\ln\left(e+\frac{\varepsilon OPT^*|\Theta^k|}{2\Delta R}\right)$ into (\ref{eq25}), we have
	\begin{align}
		\mathbb{E}[R(\vec{b},\vec{A})]&\geq OPT^*-3t\nonumber\\
		&\geq OPT^*-\frac{6\Delta R}{\varepsilon}\cdot\ln\left(e+\frac{\varepsilon OPT^*|\Theta^k|}{2\Delta R}\right)\nonumber\\
		&\geq F(\Theta)\cdot OPT-\frac{6\Delta R}{\varepsilon}\cdot\ln\left(e+\frac{\varepsilon OPT|\Theta^k|}{2\Delta R}\right).\nonumber
	\end{align}
	Therefore, Theorem \ref{thm6} has been proven.
\end{proof}

\section{Implementation and Simulation}
Shown as Theorem \ref{thm5}, the running time of Algorithm \ref{alg1} can be bounded by $|\Theta^k|$, which is not computationally efficient. Thus, in this section, we first discuss an implementation technique to reduce the time complexity to polynomial time. Then, we implement and evaluate our proposed mechanism by extensive simulations.

\subsection{Implementation Technique}
In order to reduce the running time, we can learn from the recent research in \cite{ni2021differentially} to select the unit price of each resource one by one instead of selecting the price vector. The procedure is shown in Algorithm \ref{alg2}. 

\begin{algorithm}[!t]
	\caption{\text{DPAM-S}}\label{alg2}
	\begin{algorithmic}[1]
		\renewcommand{\algorithmicrequire}{\textbf{Input:}}
		\renewcommand{\algorithmicensure}{\textbf{Output:}}
		\REQUIRE $(\{\mathcal{B}_i\}_{TD_i\in\mathbb{TD}},\{\mathcal{A}_j\}_{EN_j\in\mathbb{EN}})$, $\varepsilon$, $\Theta$
		\ENSURE $\vec{X}_{\vec{p}}$, $\vec{p}$
		\STATE Initialize $\bar{b}_i=b_i/(\sum_{z=1}^{k}d_i^z)$ for each $TD_i\in\mathbb{TD}$
		\STATE Initialize $\varepsilon'=\varepsilon/k$
		\FOR {$\ell\leftarrow 1$ to $k$}
		\STATE Initialize $\Delta R^\ell=\sum_{j=1}^{n}(c_{max}-c_{min})\cdot\sum_{z=1}^{\ell}h_j^z$
		\FOR {each $p_\ell\in\Theta$}
		\STATE Initialize $x_{ij}=0$ for each $x_{ij}\in\vec{X}_{\ell,p_\ell}$
		\STATE Initialize $\mathbb{TD}_c^\ell\leftarrow\emptyset$
		\FOR {each $TD_i\in\mathbb{TD}$}
		\IF {$\sum_{z=1}^{\ell}p_z\cdot d_i^z\leq\bar{b}_i\cdot\sum_{z=1}^{\ell}d_i^z$}
		\STATE $\mathbb{TD}_c^\ell\leftarrow\mathbb{TD}_c^\ell\cup\{TD_i\}$
		\ENDIF
		\ENDFOR
		\STATE Sort the $\mathbb{TD}_c^\ell$ s.t. $\sum_{z=1}^{k}d_1^z\geq\sum_{z=1}^{k}d_2^z\geq\cdots$
		\STATE Initialize $\{{h_j^1}',{h_j^2}',\cdots,{h_j^\ell}'\}$ where ${h_j^z}'=h_j^z\in\mathbb{H}_j$
		\FOR {each $TD_i\in\mathbb{TD}_c^\ell$}
			\STATE Initialize $\mathbb{EN}^\ell_{c,i}\leftarrow\emptyset$
			\FOR {each $EN_j\in\mathbb{EN}$}
			\IF {${h_j^z}'\geq d_i^z$ for each $r_z\in\{r_1,\cdots,r_\ell\}$, $\delta_{ij}\leq dm_i$, and $\sum_{z=1}^{\ell}(p_z-a_z)\cdot d_i^z\geq 0$}
			\STATE $\mathbb{EN}^\ell_{c,i}\leftarrow\mathbb{EN}^\ell_{c,i}\cup\{EN_j\}$
			\ENDIF
			\ENDFOR
			\IF {$\mathbb{EN}^\ell_{c,i}\neq\emptyset$}
			\STATE $EN_{j^*}\leftarrow\arg\min_{EN_j\in\mathbb{EN}^\ell_{c,i}}\{\delta_{ij}\}$
			\FOR {each $r_z\in\{r_1,\cdots,r_\ell\}$}
			\STATE ${h_{j^*}^z}'\leftarrow{h_{j^*}^z}'-d_i^z$
			\ENDFOR
			\STATE $x_{ij^*}\leftarrow 1$
			\ENDIF
		\ENDFOR
		\STATE $R^\ell((\vec{b},\vec{A}),p_\ell)=\sum_{j=1}^{n}[\sum_{z=1}^{\ell}(p_z-a_j^z)\sum_{i=1}^{m}x_{ij} d_i^z]$
		\ENDFOR
		\STATE Select a $p_\ell\in\Theta$ according to the selection distribution: 		$\Pr[M^\ell(\vec{b},\vec{A})=p_\ell]=\frac{\exp\left(\frac{\varepsilon' R^\ell((\vec{b},\vec{A}),p_\ell)}{2\Delta R^\ell}\right)}{\sum_{p_\ell'\in\Theta}\exp\left(\frac{\varepsilon' R^\ell((\vec{b},\vec{A}),p_\ell')}{2\Delta R^\ell}\right)}$
		\ENDFOR
		\RETURN $\vec{X}_{k,p_k}$, $\vec{p}=\{p_1,p_2,\cdots,p_k\}$
	\end{algorithmic}
\end{algorithm}

First of all, we define an average unit bid of each buyer $TD_i\in\mathbb{TD}$ as $\bar{b}_i=b_i/(\sum_{z=1}^{k}d_i^z)$. The main loop that iterates $k$ times to check all kinds of resources. When checking the $\ell$-th unit price $p_\ell\in\Theta$ $(1\leq\ell\leq k)$, we have known the previous first $\ell-1$ unit pricings. A partial price vector $(p_1,p_2,\cdots,p_{\ell-1})$ has been determined. Given a unit price $p_\ell\in\Theta$, we select  partial winning buyer candidates $\mathbb{TD}^\ell_c\subseteq\mathbb{TD}$ such that for each buyer $TD_i\in\mathbb{TD}^\ell_c$, we have
\begin{equation}
	\sum_{z=1}^{\ell}p_z\cdot d_i^z\leq\bar{b}_i\cdot\sum_{z=1}^{\ell}d_i^z.
\end{equation}
Then, it geneerates an allocation matrix $\vec{X}_{\ell,p_\ell}$ similar to the DPAM. The partial revenue can be calculate by
\begin{equation}
	R^\ell((\vec{b},\vec{A}),p_\ell)=\sum_{j=1}^{n}\left[\sum_{z=1}^{\ell}(p_z-a_j^z)\cdot\sum_{i=1}^{m}x_{ij}\cdot d_i^z\right].
\end{equation}
Here, the sensitivity of partial score function $R^\ell(\cdot)$ can be written as $\Delta R^\ell=\sum_{j=1}^{n}(c_{max}-c_{min})\cdot\sum_{z=1}^{\ell}h_j^z$. According to the exponential mechanism, the probability distribution of selection a unit price $p_\ell\in\Theta$ can be defined as follows.
\begin{equation}\label{}
	\Pr[M^\ell(\vec{b},\vec{A})=p_\ell]=\frac{\exp\left(\frac{\varepsilon' R^\ell((\vec{b},\vec{A}),p_\ell)}{2\Delta R^\ell}\right)}{\sum_{p_\ell'\in\Theta}\exp\left(\frac{\varepsilon' R^\ell((\vec{b},\vec{A}),p_\ell')}{2\Delta R^\ell}\right)}
\end{equation}
where $\varepsilon'=\varepsilon/k$. Therefore, the time complexity is reduced from $O((mn+m\log m)\cdot|\Theta^k|)$ to $O((mn+m\log m)\cdot k|\Theta|)$ according to Algorithm \ref{alg2}.

\begin{thm}
	The DPAM-S achieves $\varepsilon$-differential privacy, $\gamma$-truthfulness, individual rationality, budget balanced, computational efficiency. Moreover, the expected revenue $\mathbb{E}[R^k(\vec{b},\vec{A})]$ achieved by DPAM-S satisfies that $\mathbb{E}[R^k(\vec{b},\vec{A})]\geq$
	\begin{equation}
		F(\Theta)\cdot OPT-\frac{6k\Delta R}{\varepsilon}\cdot\ln\left(e+\frac{\varepsilon OPT|\Theta|}{2\Delta R}\right).
	\end{equation}
\end{thm}
\begin{proof}
	Based on Theorem \ref{thm1} to Theorem \ref{thm6} in this paper and Theorem 7 in \cite{ni2021differentially}, this theorem can be proven. 
\end{proof}

There are two mechanisms, DPAM and DPAM-S, to maximize the revenue of edge computing platform and satisfy desirable properties. Given a fixed privacy budget $\varepsilon$, the revenue achieved by DPAM is better, but the running time of DPAM-S is better. Which mechanism is better depends on the requirements between performance and running time.

\subsection{Simulation Setup}
To simulate this scenario, we construct a virtual rectangular region with $1000\times 1000$, where there are $m$ IoT devices and $n$ edge nodes distributed in this area uniformly. For each $TD_i\in\mathbb{TD}$, we define its coordinate as $(x_i,y_i)$. Similarly, we have $(x_j,y_j)$ for each $EN_j\in\mathbb{EN}$. The distance $\delta_{ij}$ between IoT device $TD_i$ and edge node $EN_j$ can be written as $\delta_{ij}=\sqrt{(x_i-x_j)^2+(y_i-y_j)^2}$. For each buyer $TD_i\in\mathbb{TD}$, its bidding information contains a maximum permitted distance $dm_i$, which is distributed in $[200\sqrt{2},1000\sqrt{2}]$ uniformly since the maximum distance between IoT devices and edge nodes is $1000\sqrt{2}$ in this area.

Suppose the price of a unit of resources can be normalized in $[0,1]$, then $\Theta=[c_{min},c_{max}]=[0,1]$. To implement our mechanisms, the first step is to discretize this intervel $[0,1]$ so as to traverse all possible price vectors in the space $\Theta^k$. Here, we define a concept called ``granularity'', denoted by $\sigma$. The $\sigma=0.02$ implies that we divide the interval $[0,1]$ equally into fifty parts, that is $\Theta=\{0,0.02,0.04,\cdots,0.98,1\}$. The granularity can be used as an effective method to balance the performance and time complexity. Since $\Theta=[0,1]$, we sample $a_j^z$ for each seller $EN_j\in\mathbb{EN}$ and $r_z\in\mathbb{R}$ uniformly in $[0,1]$. Next, we assume the number of resource types $k\in\{1,2,\cdots,5\}$, and the available range of each resource $[h_{min},h_{max}]=[0,20]$. Therefore, in the simulation, we make $h_j^z\in\mathbb{H}_j$ for each seller $EN_j\in\mathbb{EN}$ and $r_z\in\mathbb{R}$ distributed in $[10,20]$ uniformly. Similarly, we assume the $[d_{min},d_{max}]=[1,5]$, thus we make $d_i^z\in\mathbb{D}_i$ for each buyer $TD_i\in\mathbb{TD}$ and $r_z\in\mathbb{R}$ distributed in $[1,5]$ uniformly.

In the next step, we need to discuss how buyers decide their total bids. That is, how can we sample a total bid $b_i$ for each buyer $TD_i\in\mathbb{TD}$. According to our preceding description, the average price per unit resource is $(1+0)/2=0.5$. Here, we point out a reasonable assumption that the total bid is related to the total demand of the buyer for resources. Thus, we can sample the $b_i$ for each buyer $TD_i\in\mathbb{TD}$ as follows:
\begin{equation}
	b_i=(0.5)\cdot\sum_{z=1}^{k}d_i^z\cdot U(0.7,1.3)
\end{equation}
where the $U(0.7,1.3)$ is a value sampled from the interval $[0.7,1.3]$ uniformly.

Due to the introduction of differential privacy, the auction results have certain randomness. Thus, given a mechanism, its result is the average value of $500$ trials. To analyze the performance of our mechanisms based on differential privacy, we need a reference. For example in line 32 of Algorithm \ref{alg1}, we select a $\vec{p}\in\Theta^k$ such that maximizing $R((\vec{b},\vec{A}),\vec{p})$ as the final result. By removing the randomness (differential privacy) of the DPAM and DPAM-S, we can define two deterministic auction mechanisms, marked by ``DTAM'' and ``DTAM-S'', as references. Here, ``DT'' implies ``deterministic''. Finally, we select three typical metrics to evaluate the performance of our proposed mechanisms, which are shown as follows.
\begin{enumerate}
	\item Expected revenue of edge computing platform: it can be computed by (\ref{eq18}).
	\item Expected satisfaction: the ratio of the number of satisfied IoT devices to the total number of IoT devices.
	\item Running time: the time taken to execute a trial.
\end{enumerate}

\subsection{Simulation Results and Analysis}

\begin{figure*}[!t]
	\centering
	\subfigure[(Expected) Revenue]{
		\includegraphics[width=0.325\linewidth]{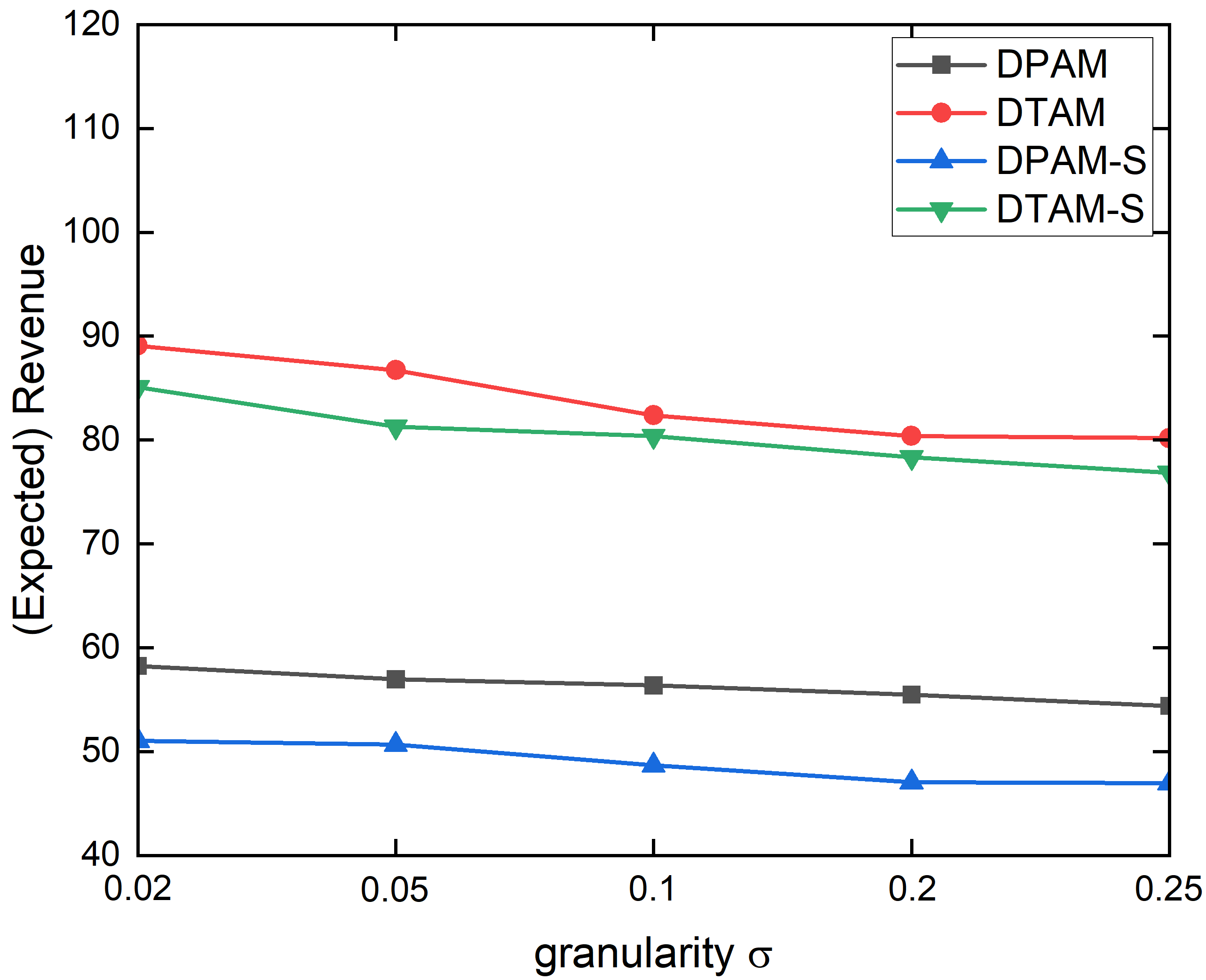}
		%\caption{fig1}
	}%
	\subfigure[(Expected) Satisfaction]{
		\includegraphics[width=0.325\linewidth]{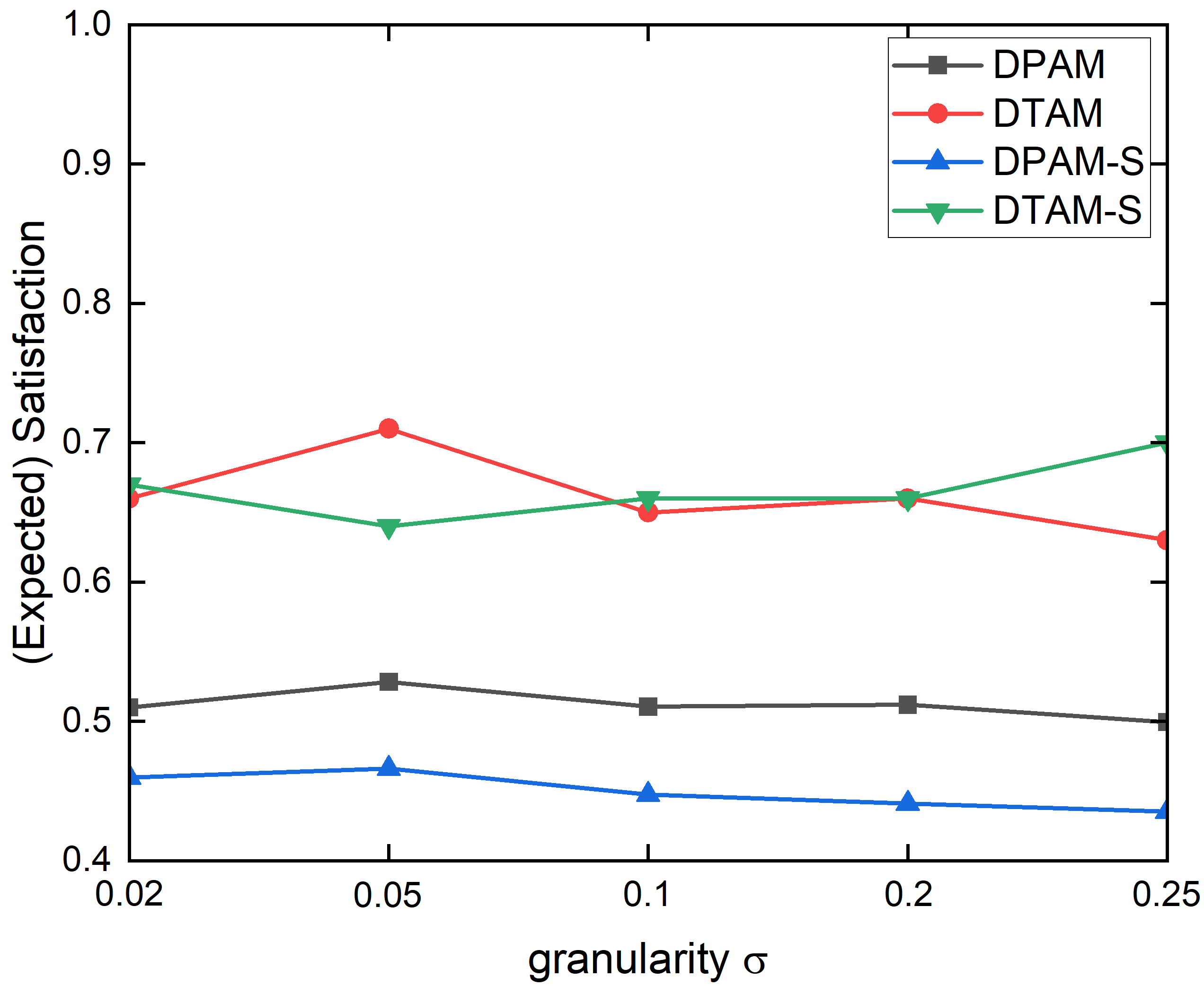}
		%\caption{fig1}
	}%
	\subfigure[Running time]{
		\includegraphics[width=0.325\linewidth]{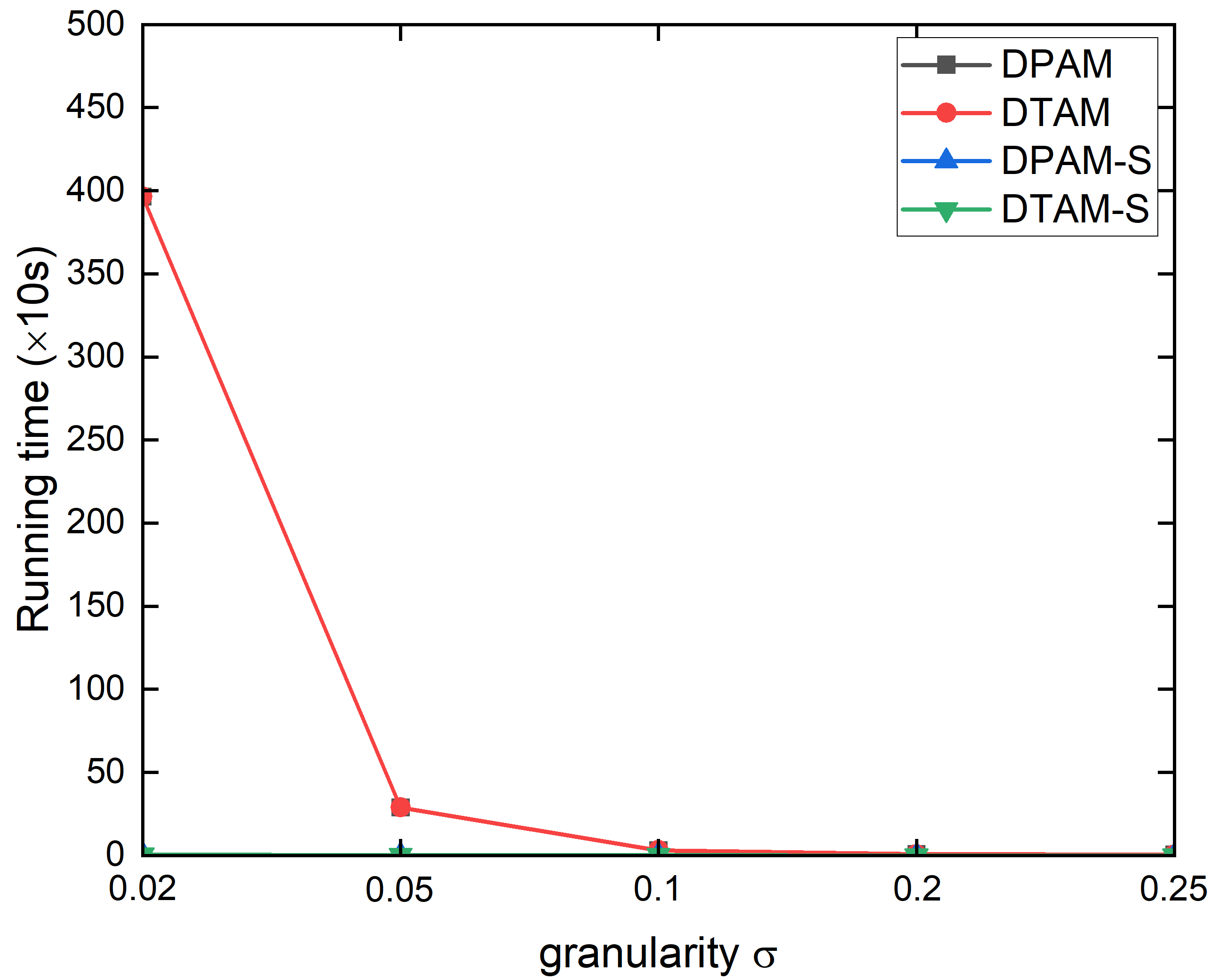}
		%\caption{fig1}
	}%
	\caption{The performances of proposed mechanisms on different granularities, where $m=100$, $n=50$, $k=3$, and $\varepsilon=200$.}
	\label{fig2}
\end{figure*}

\begin{figure*}[!t]
	\centering
	\subfigure[(Expected) Revenue]{
		\includegraphics[width=0.323\linewidth]{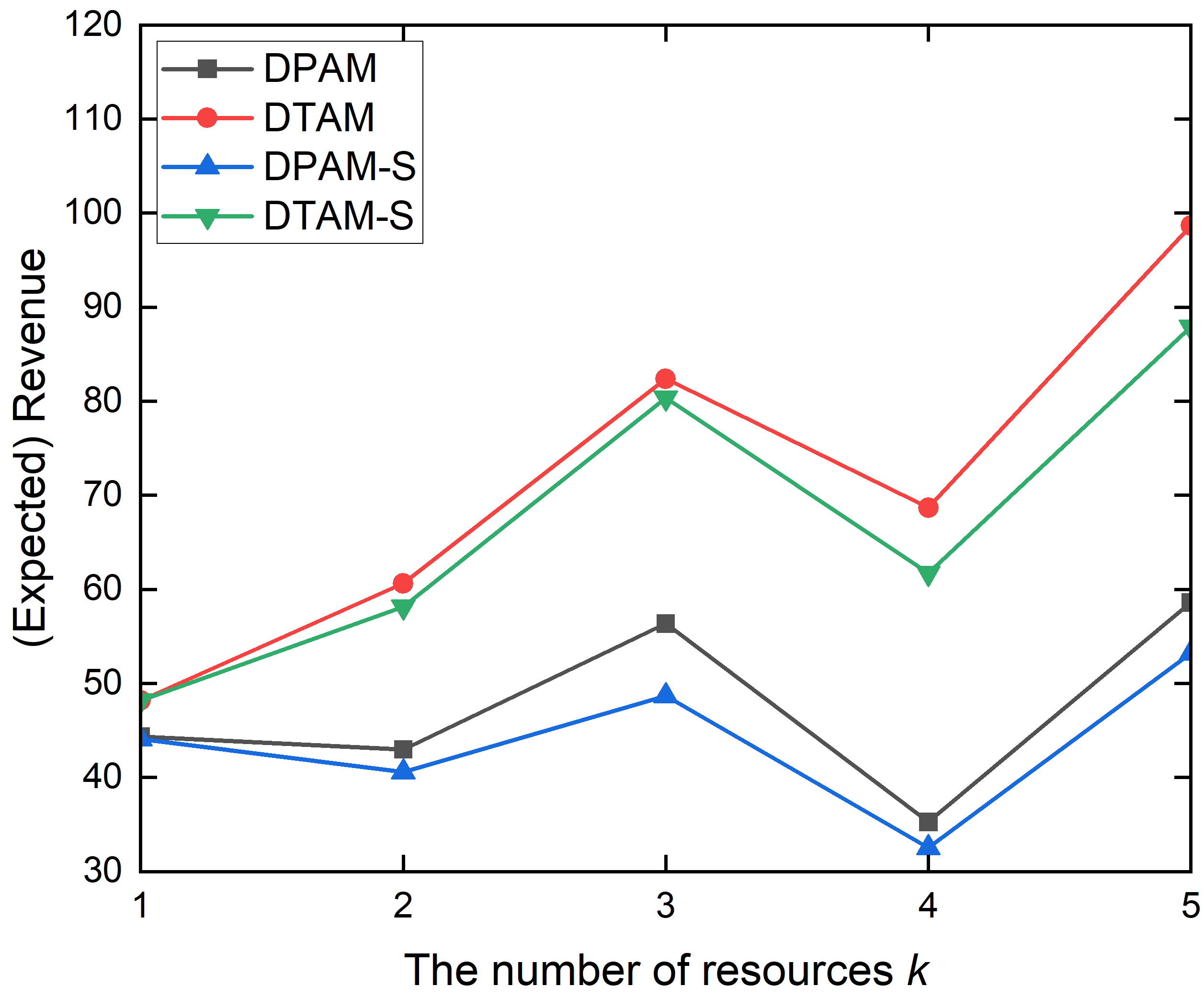}
		%\caption{fig1}
	}%
	\subfigure[(Expected) Satisfaction]{
		\includegraphics[width=0.323\linewidth]{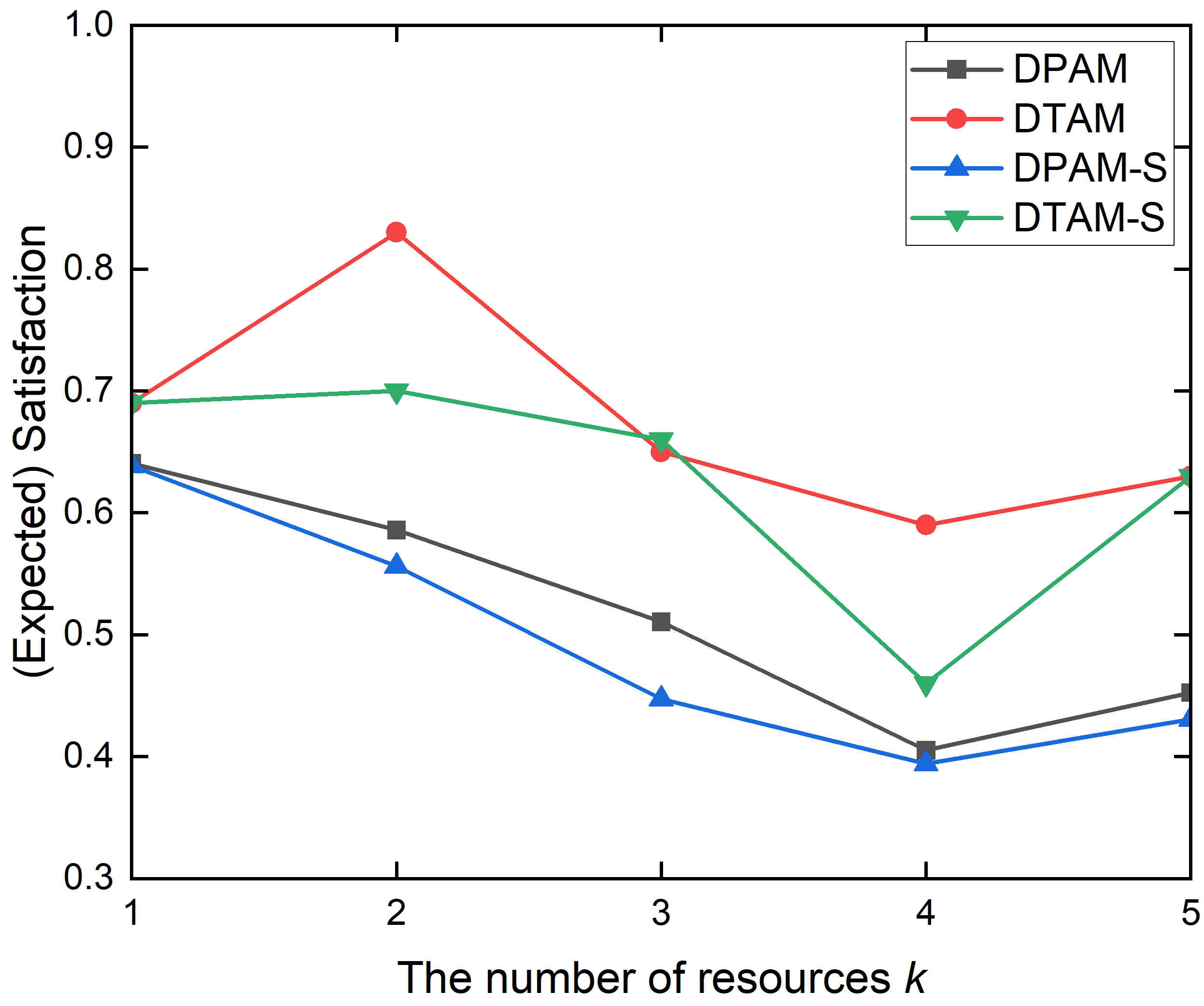}
		%\caption{fig1}
	}%
	\subfigure[Running time]{
		\includegraphics[width=0.323\linewidth]{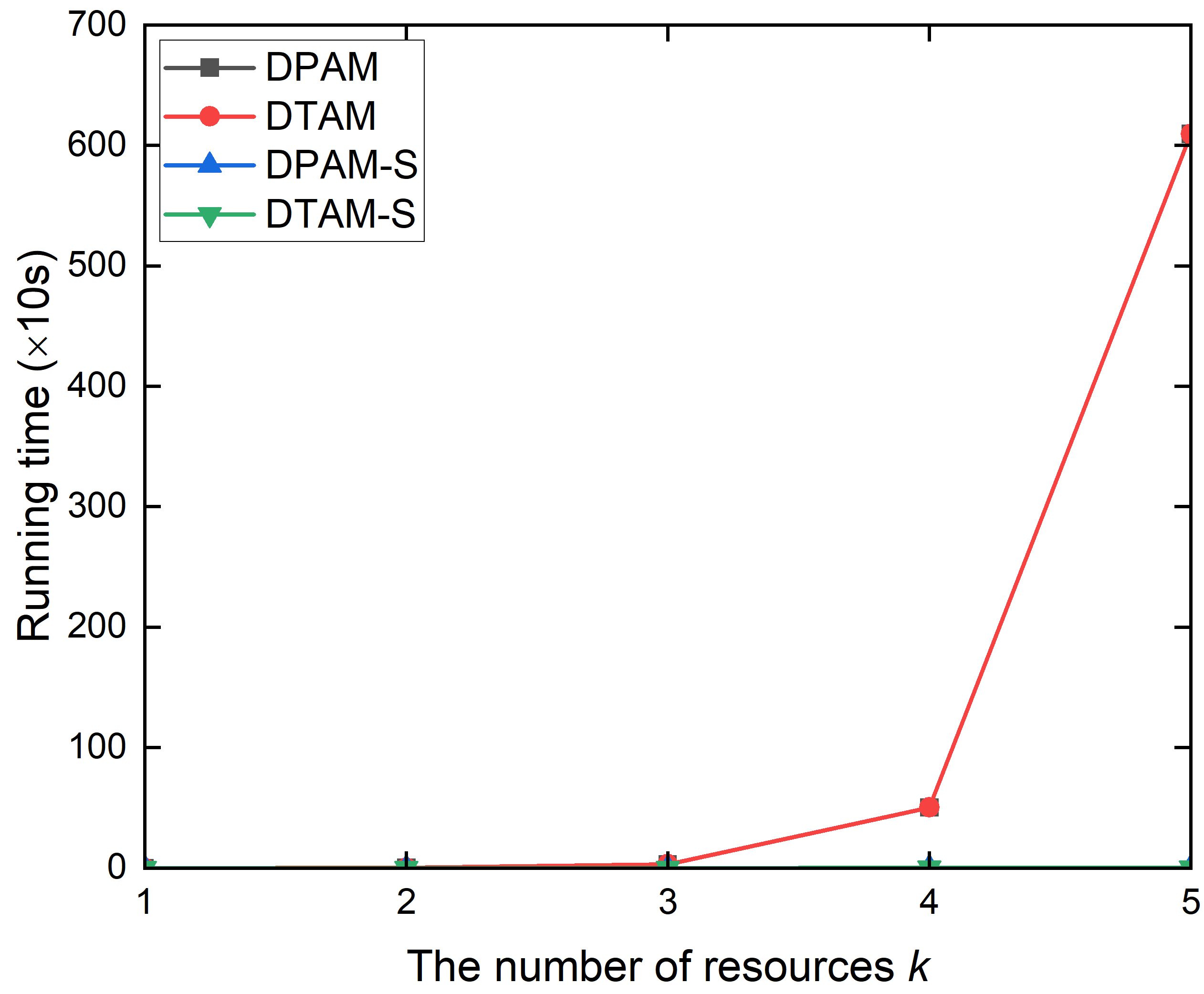}
		%\caption{fig1}
	}%
	\caption{The performances of proposed mechanisms on different number of resourse types, where $m=100$, $n=50$, $\sigma=0.1$, and $\varepsilon=200$.}
	\label{fig3}
\end{figure*}

In any time slot $t\in T$, there are $m$ IoT devices (buyers) and $n$ edge nodes. Generally, the number of IoT devices is much larger than that of edge nodes, thus we assume that $m\geq n$ in our following simulations. Our task in this part can be divided into four parts, which discuss the impact of granularity, the number of resource types, the number of IoT devices, and privacy budget on the performance of our proposed mechanisms respectively.

\textbf{Granularity: }Figure \ref{fig2} plots the revenues, satisfactions, and running times of four different mechanisms vary with the increase of granularity, where we assume $m=100$, $n=10$, $k=3$, and $\varepsilon=200$. Shown as Figure \ref{fig2} (a), we can see that the revenue will decrease slightly with the increase of granularity. This is because smaller granularity means higher accuracy, thus we can compute more price vectors and select the better one. Similar results are also reflected in users' satisfaction. Shown as Figure \ref{fig2} (c), the running time will increase significantly with the decrease of granularity. Here, let us make a rough analysis. Supposing $\sigma_1=0.1\cdot\sigma_2$, we have $|\Theta_1|=10\cdot|\Theta_2|$. In the DPAM (DTAM),  the running time under the granularity $\sigma_1$ is $10^k$ times as much as that under the granularity $\sigma_2$. And in the DPAM-S (DPTM-S), the running time under the granularity $\sigma_1$ is $10\cdot k$ times as much as that under the granularity $\sigma_2$. The simulation results in Figure \ref{fig3} (c) meets our expectations in general. We have mentioned that the granularity is a method to balance performance and time complexity. Based on the results of Figure \ref{fig2}, we set the granularity $\sigma=0.1$ in the following simulations.

\textbf{The number of resource types: }Figure \ref{fig3} plots the performances vary with the increasing number of resource types, where we assume $m=100$, $n=10$, $\sigma=0.1$, and $\varepsilon=200$. Shown as Figure \ref{fig3} (a), we observe that the revenue will show an upward trend with the increase of resource types. This is because the increase in resource types enables each edge nodes to sell more resource units. However, there is an exception when $k=4$. From Figure \ref{fig3} (b), we can see that the users' satisfaction drops obviously when $k=4$. This may be due to the randomness of data, which makes the resource request of IoT devices difficult to realize, which leads to the decline of their satisfaction. In addition, another important discovery is that the gap between DPAM (DPAM-S) and DTAM (DTAM-S) increases with the increase of resource types. Under a larger $k$, the sample space $\Theta^k$ will become larger, resulting in higher randomness. In other words, the probability of choosing the optimal solution will become smaller. Shown as Figure \ref{fig3} (c), the running time will increase significantly with the increasing number of resource types. Similarly, we suppose $k_1=k_2+1$. In the DPAM (DTAM), the running time under the $k_1$ is $|\Theta|$ times as much as that under the $k_2$ since we have  $|\Theta|^{k_1}=|\Theta|^{k_2}\cdot|\Theta|$. And in the DPAM-S (DPTM-S), the running time under the $k_1$ is $(k_2+1)/k_2$ times as much as that under the $k_2$. If there are a large number of resource types, the DPAM (DTAM) is undisirable since its running time grows exponentially. By contrast, the running time of DPAM-S (DTAM-S) grows linearly.

\begin{figure*}[!t]
	\centering
	\subfigure[(Expected) Revenue]{
		\includegraphics[width=0.325\linewidth]{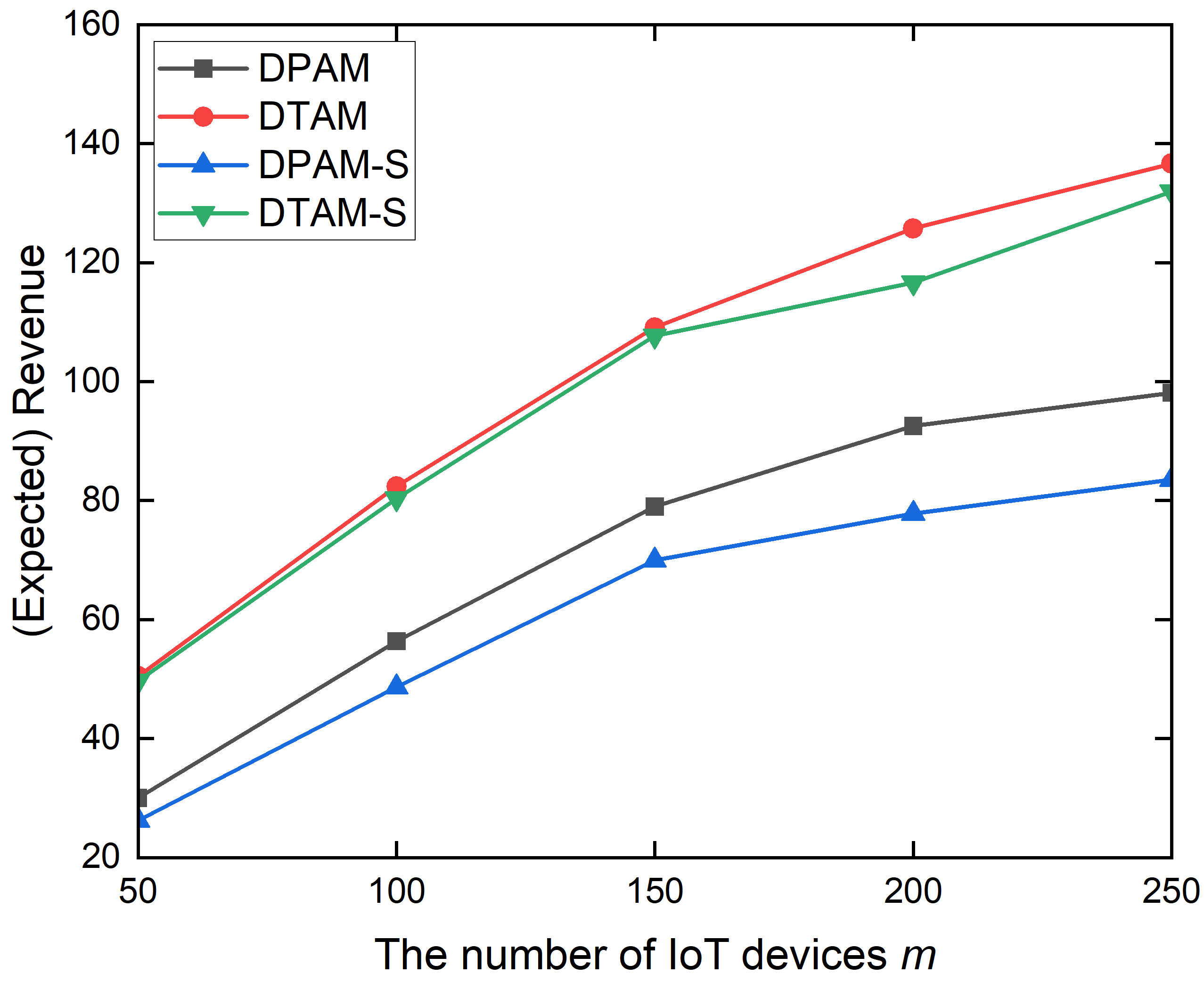}
		%\caption{fig1}
	}%
	\subfigure[(Expected) Satisfaction]{
		\includegraphics[width=0.325\linewidth]{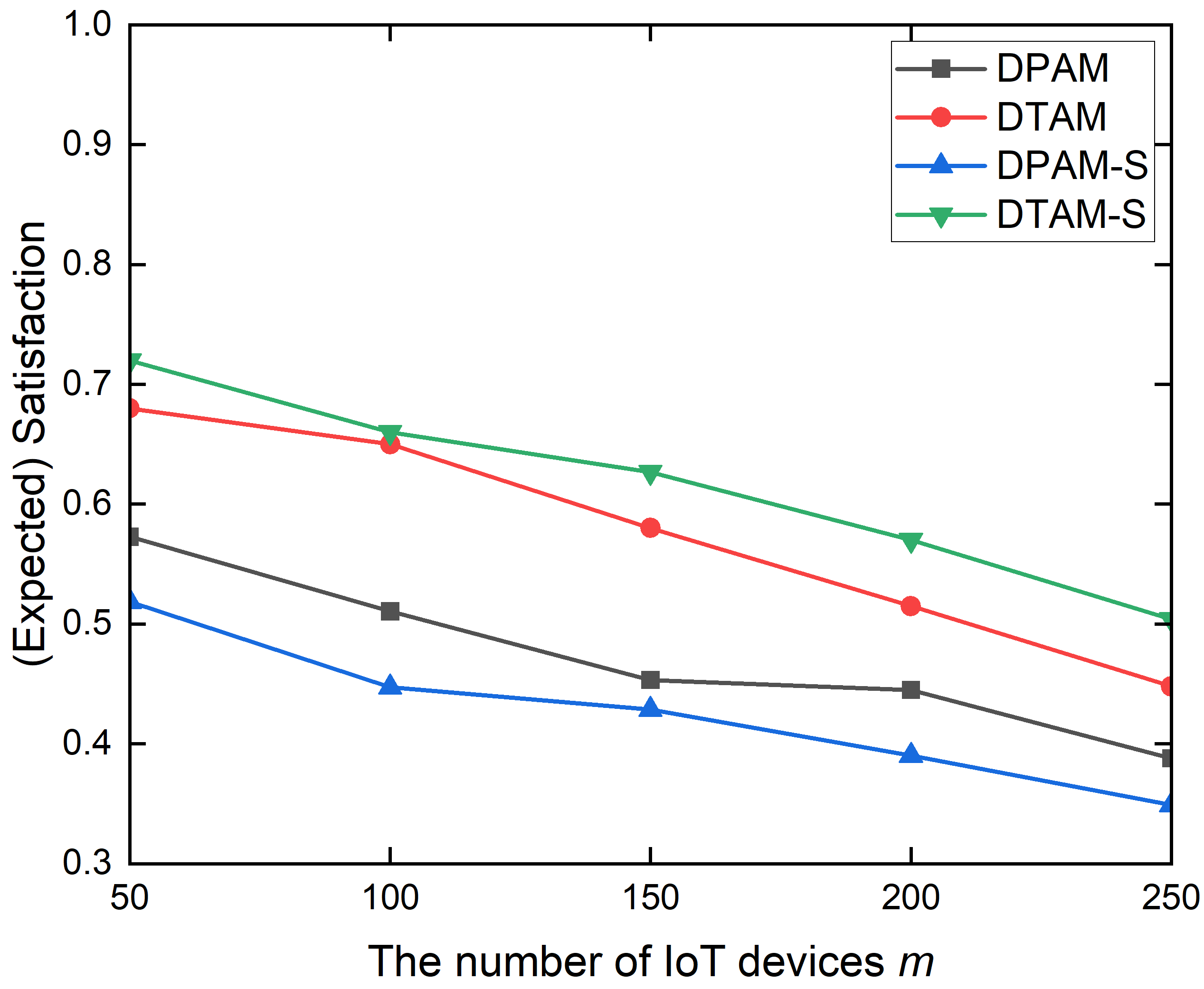}
		%\caption{fig1}
	}%
	\subfigure[Running time]{
		\includegraphics[width=0.325\linewidth]{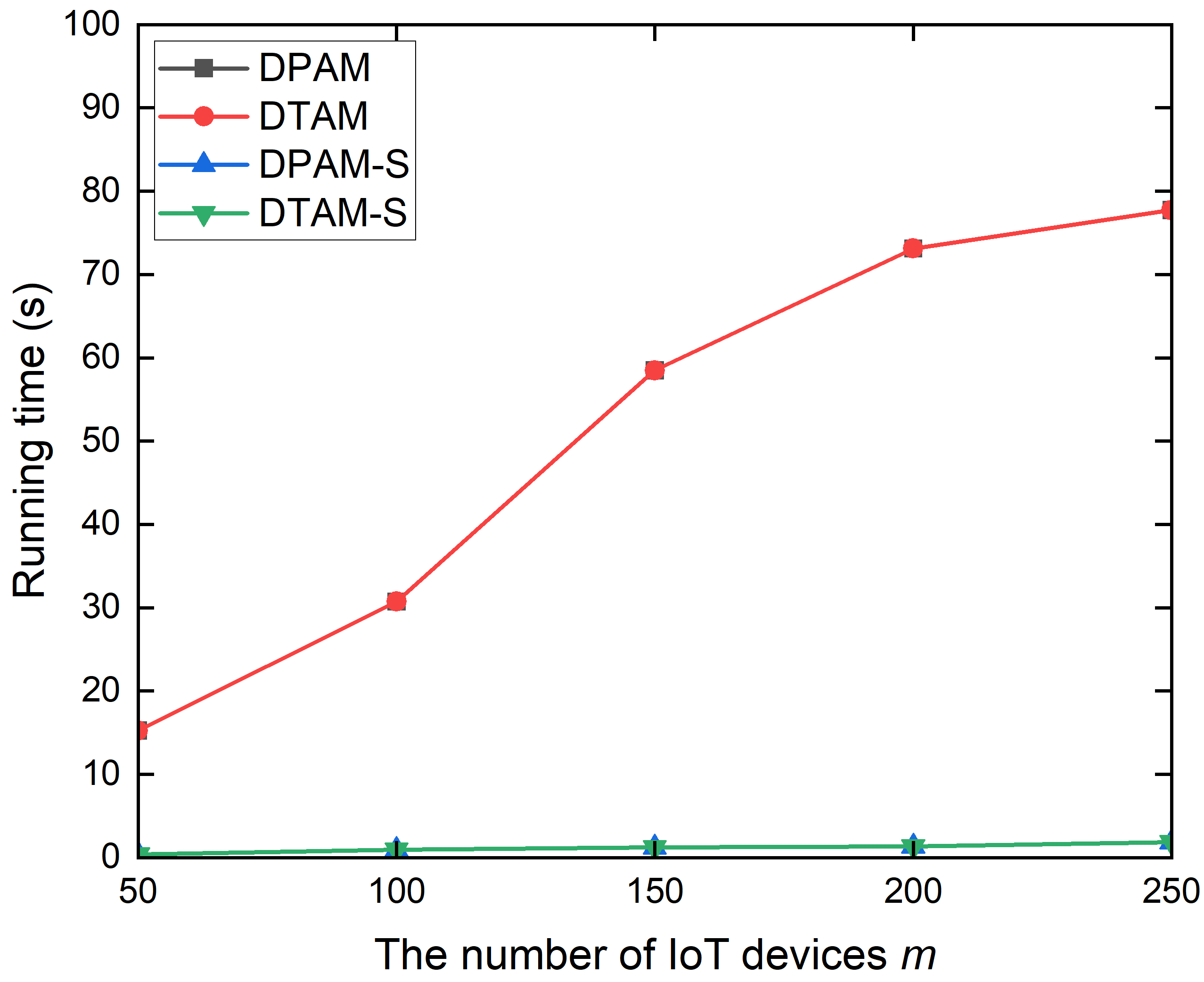}
		%\caption{fig1}
	}%
	\caption{The performances of proposed mechanisms on different number of edge nodes, where $n=50$, $k=3$, $\sigma=0.1$, and $\varepsilon=200$.}
	\label{fig4}
\end{figure*}

\begin{figure*}[!t]
	\centering
	\subfigure[(Expected) Revenue]{
		\includegraphics[width=0.325\linewidth]{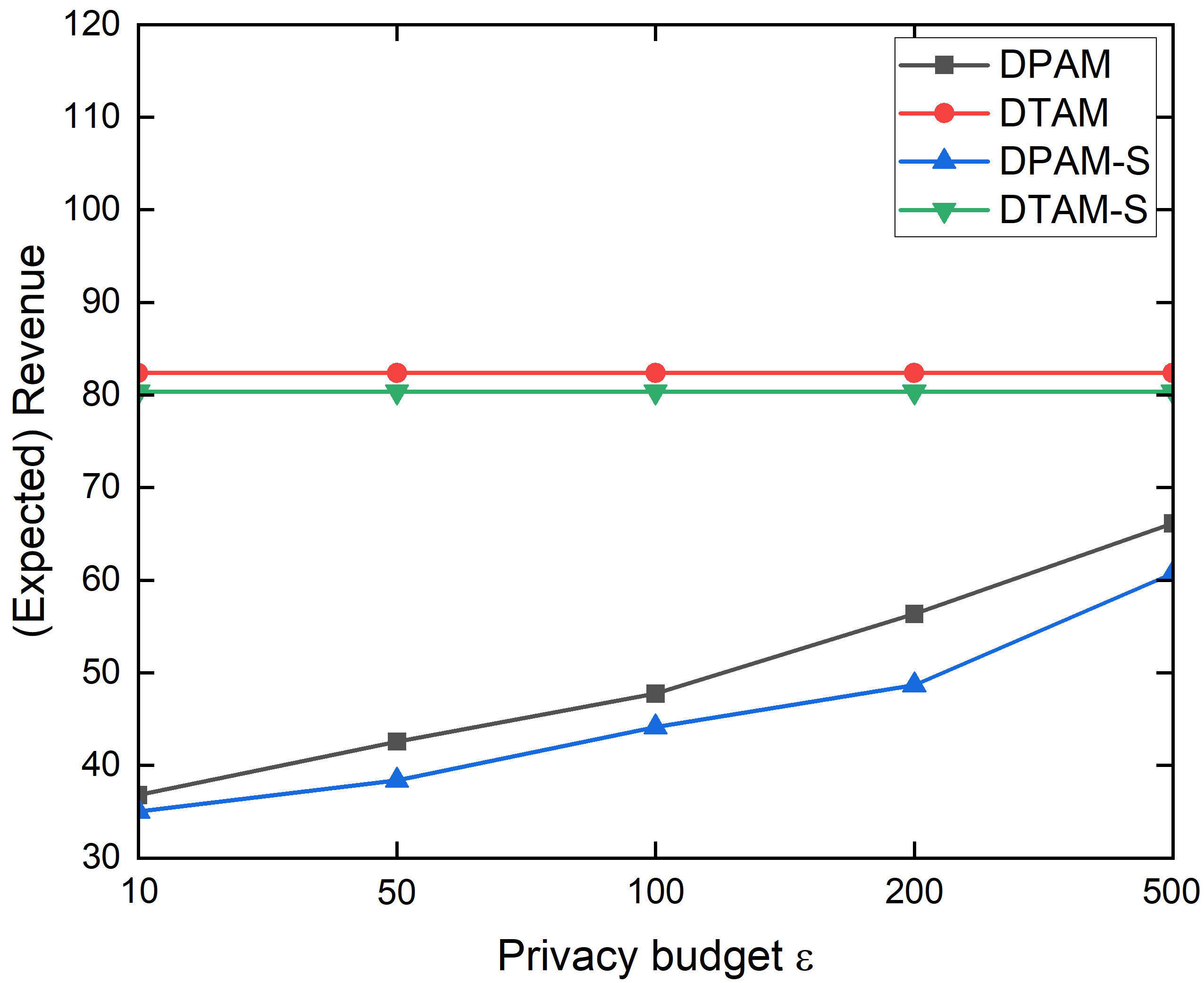}
		%\caption{fig1}
	}%
	\subfigure[(Expected) Satisfaction]{
		\includegraphics[width=0.325\linewidth]{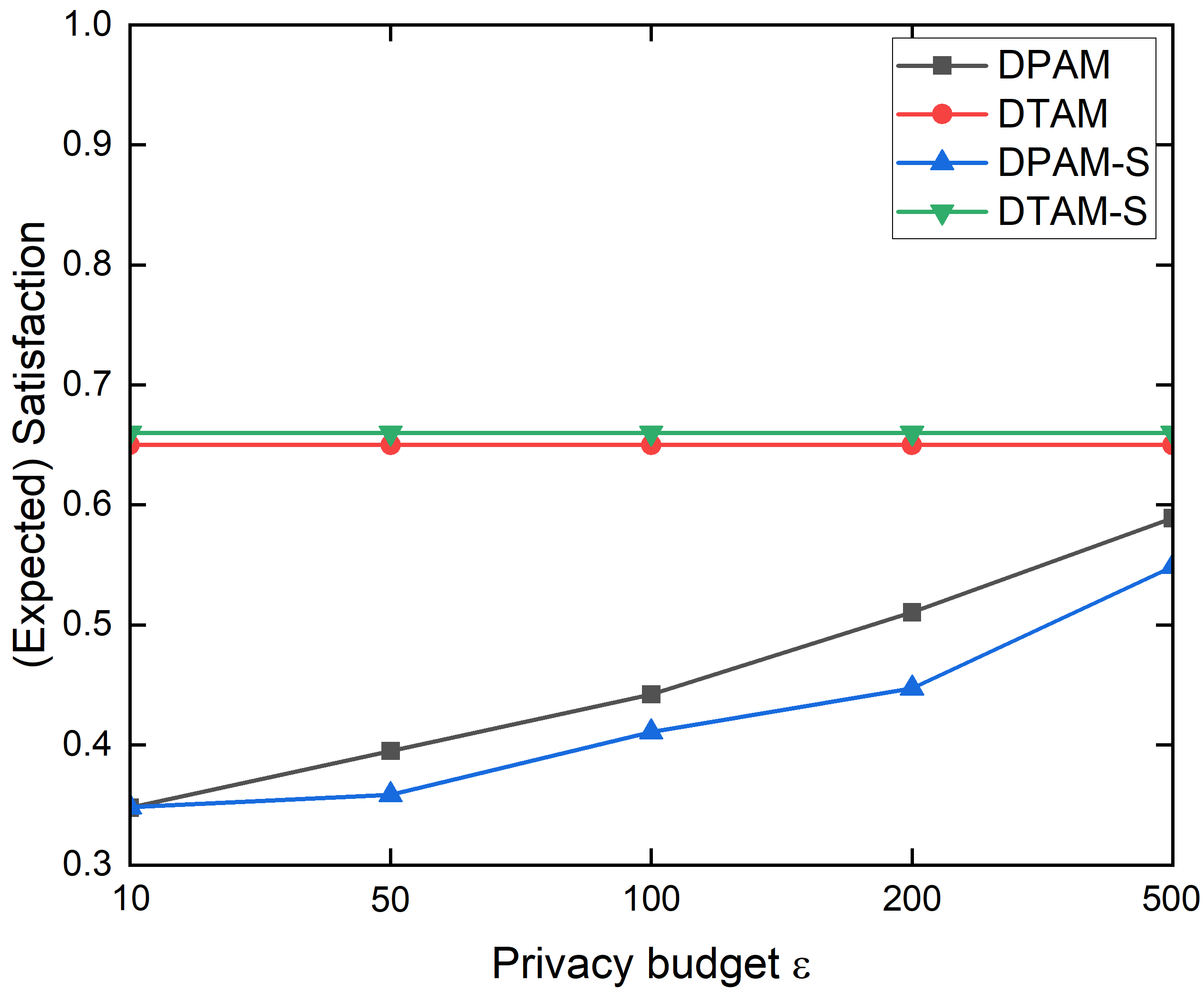}
		%\caption{fig1}
	}%
	\subfigure[Running time]{
		\includegraphics[width=0.325\linewidth]{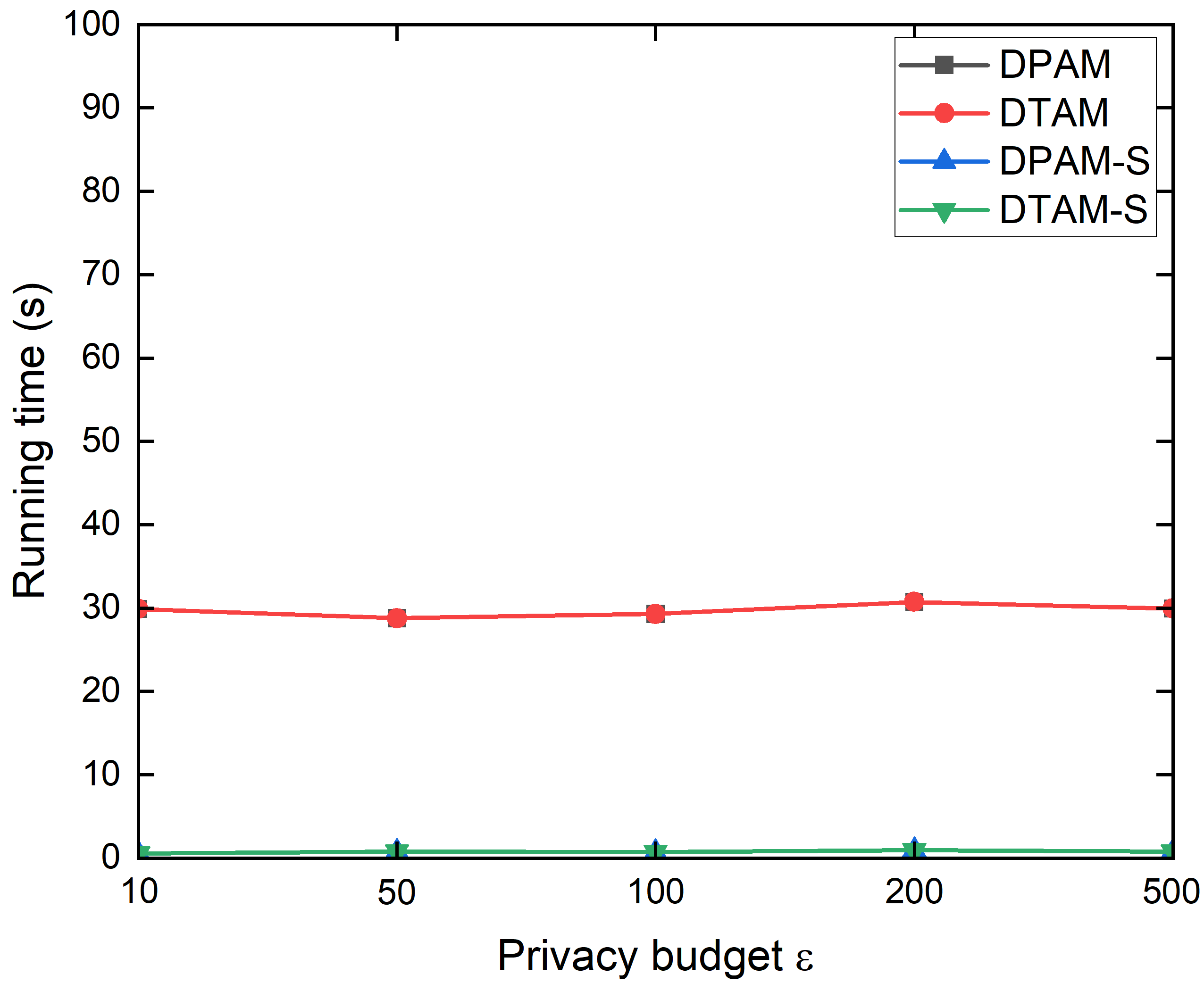}
		%\caption{fig1}
	}%
	\caption{The performances of proposed mechanisms on different privacy budgets, where $m=100$, $n=50$, $k=3$, and $\sigma=0.1$.}
	\label{fig5}
\end{figure*}

\textbf{The number of IoT devices: }Figure \ref{fig4} plots the revenues, satisfactions, and running times of four different mechanisms vary with the increasing number of IoT devices, where we assume $n=10$, $k=3$, $\sigma=0.1$, and $\varepsilon=200$. Shown as Figure \ref{fig4} (a) and (b), we can see that the revenue will increase and the satisfaction will decrease with the increasing number of IoT devices. This is because there are more feasible buyer candidates requesting resources, so that the resources of edge nodes can be more fully utilized. Although more IoT devices can be satisfied, the total number of IoT devices becomes much more, resulting in a decline in satisfaction. Shown as Figure \ref{fig4} (c), the running time will grow linearly with the increasing number of IoT devices, which meets our expectations in general.

\textbf{Privacy budget: }Figure \ref{fig5} plots the performances vary with the increase of privacy budget, where we assume $m=100$, $n=10$, $k=3$, and $\sigma=0.1$. Shown as Figure \ref{fig5} (a) and (b), we observe that the revenue and satisfaction remain unchanged in the DTAM and DTAM-S since they are deterministic mechanisms and have nothing to do with the value of privacy budget. In the DPAM and DPAM-S, the revenue and satisfaction show upward trends with the increase of privacy budget. Actually, the privacy budget controls the degree of protection provided by differetial privacy. The higher the privacy budget, the higher the revenue and satisfaction, but the degree of privacy protection will be weakened. Shown as Figure \ref{fig5} (c), the running time remains the same with the increase of privacy budget, which indicates that the running time has no concern with the choice of privacy budget.

Based on the above four tasks, the main conclusions can be summarized as follows. The granularity affects the running time significantly, and it is usually not necessary to choose a very small granularity to ensure accuracy. In the case of a large number of resource types, the DPAM (DTAM) is not applicable due to the limitation of time complexity. Under the condition of sufficient network bandwidth, the more participating IoT devices, the better the revenue. We need to balance the contradiction between privacy protection and revenue by choosing a privacy budget.

\section{Conclusion}
In this paper, we propose an edge-thing system based on blockchain technology and smart contract, which achieves complete decentralization and tampering-proof. In order to model the resources allocation and pricing between IoT devices and edge nodes, we formulate a novel combinatorial double auction problem. Then, we introduce differential privacy into the auction so as to prevent privacy leakage further. First, we design the DPAM mechanism, and prove it satisfies $\varepsilon$-differential privacy, $\gamma$-truthfulness, individual rationality, budget balance, but not computational efficiency. It is not suitable to use in the case of too many resource types. Then, we propose the DPAM-S mechanism to reduce the time complexity to polynomial time, and satisfy the above desired properties as well. Finally, we built a virtual region to test our proposed mechanisms by extensive simulations, which confirms our theoretical analysis.

\section*{Acknowledgment}

This work is supported by Guangdong Key Lab of AI and Multi-modal Data Processing, National Natural Science Foundation of China (NSFC) Project No. 61872239; BNU-UIC Institute of Artificial Intelligence and Future Networks funded by Beijing Normal University at Zhuhai (BNU Zhuhai) and AI-DS Research Hub, BNU-HKBU United International College (UIC), Zhuhai, Guangdong, China.

\ifCLASSOPTIONcaptionsoff
  \newpage
\fi

\bibliographystyle{IEEEtran}
\bibliography{references}

\begin{IEEEbiography}[{\includegraphics[width=1in,height=1.25in,clip,keepaspectratio]{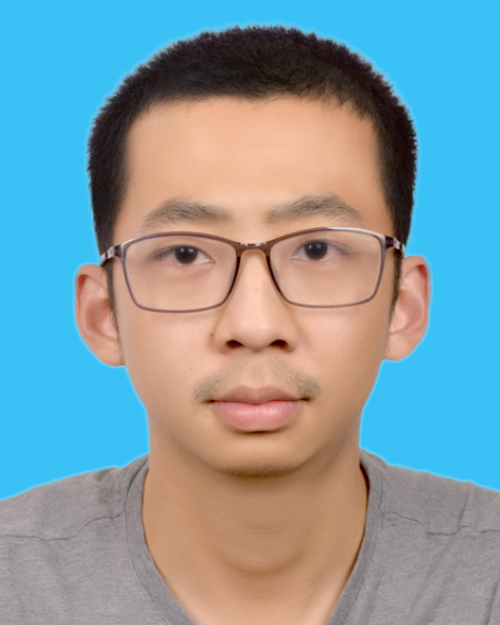}}]{Jianxiong Guo}
	received his Ph.D. degree from the Department of Computer Science, University of Texas at Dallas, Richardson, TX, USA, in 2021, and his B.E. degree from the School of Chemistry and Chemical Engineering, South China University of Technology, Guangzhou, Guangdong, China, in 2015. He is currently an Assistant Professor with the BNU-UIC Institute of Artificial Intelligence and Future Networks, Beijing Normal University at Zhuhai, and also with the Guangdong Key Lab of AI and Multi-Modal Data Processing, BNU-HKBU United International College, Zhuhai, Guangdong, China. His research interests include social networks, algorithm design, data mining, IoT application, blockchain, and combinatorial optimization.
\end{IEEEbiography}

\begin{IEEEbiography}[{\includegraphics[width=1in,height=1.25in,clip,keepaspectratio]{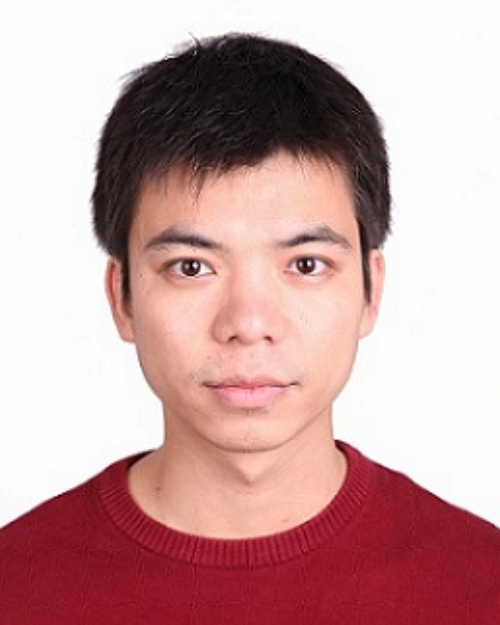}}]{Xingjian Ding}
	received his B.E. degree in electronic information engineering from Sichuan University in 2012 and M.S. degree in software engineering from Beijing Forestry University in 2017. He obtained his Ph.D. degree from the School of Information, Renmin University of China in 2021. He is currently an assistant professor at the School of Software Engineering, Beijing University of Technology. His research interests include wireless rechargeable sensor networks, approximation algorithms design and analysis, and blockchain.
\end{IEEEbiography}

\begin{IEEEbiography}[{\includegraphics[width=1in,height=1.25in,clip,keepaspectratio]{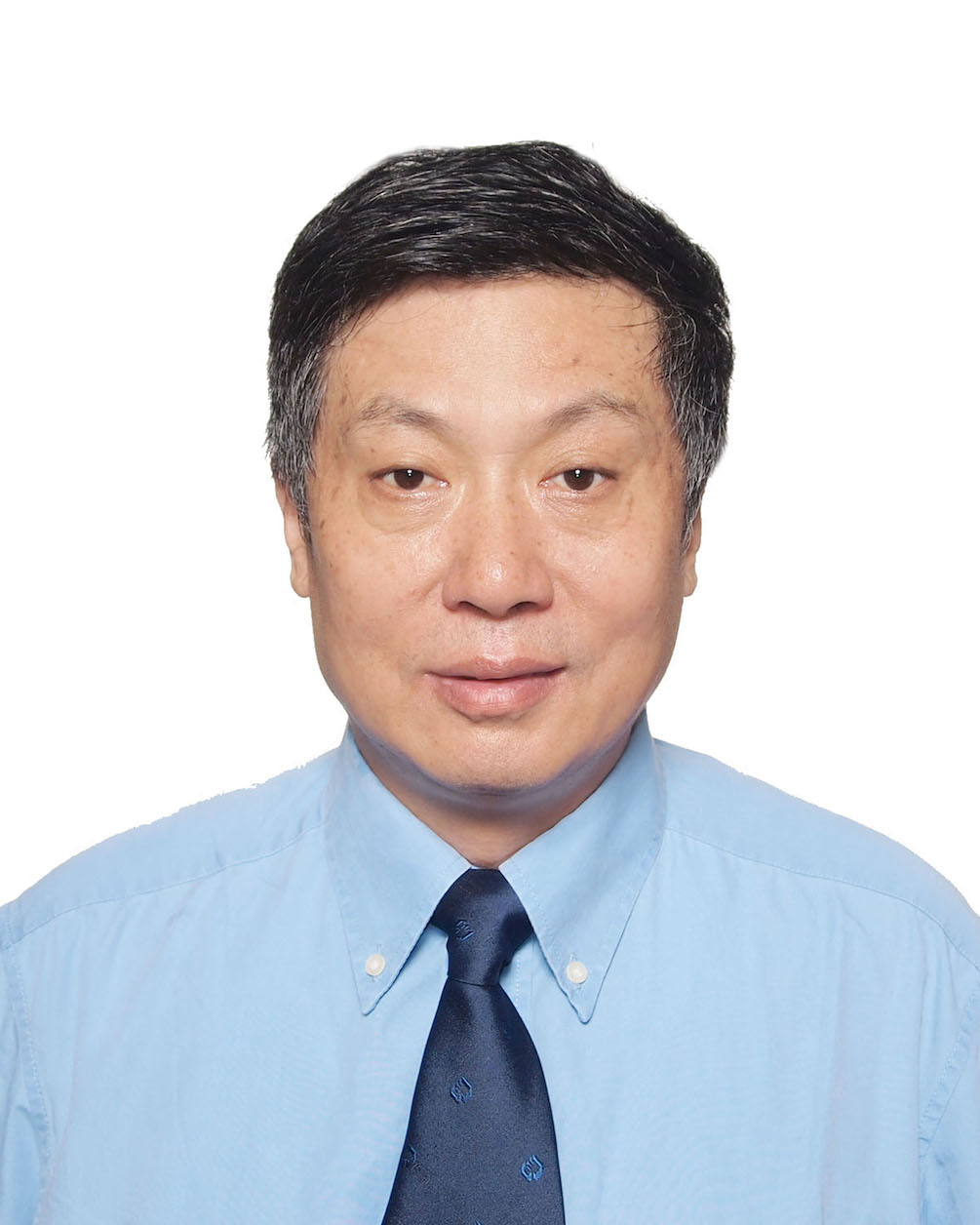}}]{Weijia Jia}
	 is currently a Chair Professor and Director of BNU-UIC Institute of Artificial Intelligence and future Networks, Beijing Normal University at Zhuhai; VP for Research of BNU-HKBU United International College. His contributions have been recognized as optimal network routing and deployment, anycast and QoS routing, sensors networking, AI (knowledge relation extractions; NLP etc.) and edge computing. He has over 600 publications in the prestige international journals/conferences and research books and book chapters. He is the Fellow of IEEE and the Distinguished Member of CCF.
\end{IEEEbiography}
\end{document}